\documentclass[a4paper,fleqn]{cas-sc}

\usepackage[numbers,sort&compress]{natbib}

\usepackage{mathtools}
\usepackage{bm}

\newtheorem{theorem}{Theorem}

\newtheorem{proposition}[theorem]{Proposition}

\newdefinition{remark}[theorem]{Remark}
\newproof{proof}{Proof}

\newcommand{\alttext}[1]{}

\newcommand{\figph}[2]{%
  \IfFileExists{#1}%
    {\includegraphics[width=\linewidth]{#1}}%
    {\setlength{\fboxsep}{0pt}%
     \fbox{\parbox[c][#2][c]{\dimexpr\linewidth-2\fboxrule\relax}{%
       \centering\color{gray}\ttfamily\scriptsize
       \detokenize{#1}\\[3pt]
       \textrm{\itshape figure placeholder}}}}%
}

\newcommand{\R}{\mathbb{R}}
\newcommand{\N}{\mathbb{N}}
\newcommand{\Om}{\Omega}
\newcommand{\Es}{E^{*}}
\newcommand{\xs}{x^{*}}
\newcommand{\ys}{y^{*}}
\newcommand{\ip}[2]{\langle #1,#2\rangle}
\newcommand{\dd}{\mathrm{d}}
\newcommand{\ii}{\mathrm{i}}

\begin{document}
\let\WriteBookmarks\relax
\def\floatpagepagefraction{1}
\def\textpagefraction{.001}

\shorttitle{Turing-Hopf bifurcation with additional food and predator competition}
\shortauthors{A. Hazra et al.}

\title[mode=title]{Predator self-limitation controls pattern formation in a predator–prey system with additional food: a Turing–Hopf analysis}

\author[1]{Anushree Hazra}%{anushree25300@gmail.com}
\credit{Formal analysis, Investigation, Software, Writing -- original draft}
\affiliation[1]{organization={Department of Mathematics, Ramakrishna Mission
            Vivekananda Educational and Research Institute},
            addressline={Belur Math},
            city={Howrah},
            postcode={711202},
            state={West Bengal},
            country={India}}

\author[2]{Aniket Banerjee}[orcid=0009-0007-6558-9708]
\cormark[1]
\ead{aniket.banerjee@sorbonne-universite.fr}
\credit{Conceptualization, Methodology, Software, Validation,
Writing -- original draft, Writing -- review \& editing}
\affiliation[2]{organization={Sorbonne Universit\'e, CNRS, Laboratoire
            Jacques-Louis Lions (LJLL)},
            addressline={4 place Jussieu},
            city={Paris},
            postcode={75005},
            country={France}}

\author[1]{Debaldev Jana}%{debaldevjana.jana@gmail.com}
\credit{Conceptualization, Supervision, Writing -- review \& editing}

\cortext[1]{Corresponding author}

%% ---------------------------------------------------------------
%%  Abstract: 245 words (CSF limit is 250).
%% ---------------------------------------------------------------
\begin{abstract}
Supplying a released predator with additional, non-reproducing food is a standard
lever in augmentative biological control, with a known drawback: with
nothing limiting the predator's own numbers, the extra food lets its population
grow without bound. Competition among the predators supplies the missing brake.
How this self-limitation reshapes the spatial arrangement of the two species has
not been asked. We address it with a reaction--diffusion model of a logistically
growing prey and a predator feeding through a Holling type-II response that also
draws on additional food of quantity $\xi$ and quality $1/\alpha$, the predators
competing among themselves at strength $c\xi$. In the well-mixed setting we
locate the Hopf bifurcation of the coexistence state exactly and show the cycle
born there is stable, so weak competition gives boom--bust oscillations, not runaway growth. Allowing movement, we obtain the diffusion-driven Turing
threshold at which the uniform state breaks into stationary patches of high and
low density, and find the uniform oscillation stable as it appears. With prey
mobility and competition strength as control parameters, the pattern-forming and
oscillatory instabilities meet at a single point, where we compute the dynamics.
Simulations confirm the sequence: weak competition gives a whole-field
oscillation, stronger competition with faster prey spread gives fixed patterns,
and near the crossover the two combine into patterns that pulse in time. Predator self-competition therefore sets the spatial structure of the
community, which is what matters when additional food is used to steer a control
agent in the field.
\end{abstract}

%% ---------------------------------------------------------------
%%  Highlights.  CSF also wants these as a SEPARATE editable file with
%%  "highlights" in the file name; see Highlights_CSF.tex.
%%  3-5 bullets, each at most 85 characters including spaces.
%% ---------------------------------------------------------------
%\begin{highlights}
%\item Predator self-limitation controls spatial, not just temporal, dynamics.
%\item The Hopf bifurcation of the kinetics is located exactly and is supercritical.
%\item Turing and Hopf curves cross transversally at an explicit codimension-2 point.
%\item The three-dimensional normal form gives a supercritical Hopf-pitchfork.
%\item Raising predator competition stabilises time dynamics but not the Turing mode.
%\end{highlights}

%% ---------------------------------------------------------------
%%  Keywords: 7, the CSF maximum.  None contains "and" or "of".
%% ---------------------------------------------------------------

\maketitle
%% =====================================================================
%%  SHARED BODY -- identical in the Elsevier and Springer versions.
%% =====================================================================

\section{Introduction}\label{sec:intro}

Two species that interact chemically or ecologically and that move at different
speeds can abandon a uniform distribution and settle into a fixed arrangement of
crowded and empty patches, without any help from the environment.
Turing~\cite{turing1952chemical} showed why, and Segel and
Jackson~\cite{segel1972dissipative} pointed out that predator and prey do this
naturally: prey build up locally and act as the activator, predators spread
faster and act as the inhibitor. Field ecologists recognise the outcome.
Semi-arid vegetation forms regular bands and gaps~\cite{klausmeier1999regular};
mound fields, mussel beds and tussocks are spaced far more regularly than chance
would give~\cite{rietkerk2008regular,pringle2017spatial}; plankton is patchy at
scales that physical forcing does not explain~\cite{medvinsky2002spatiotemporal}.
Self-organisation of this kind is now read as a common state of ecological
systems rather than a curiosity~\cite{rietkerk2008regular,pringle2017spatial}.
Dispersal has also been tied to patchy invasion when the prey suffers an Allee
effect~\cite{petrovskii2002allee,morozov2006spatiotemporal}, to the spatial
dynamics of prey that defend themselves in
groups~\cite{venturino2013spatiotemporal}, and to the onset of spatiotemporal
chaos~\cite{petrovskii1999minimal,hu2015pattern}; cross-diffusion, harvesting and
seasonally varying dispersal widen the range of outcomes
further~\cite{sambath2018spatiotemporal,bhunia2023study,tao2021study}.

The kinetics used in this literature descend from Lotka and
Volterra~\cite{lotka1925elements,volterra1926variazioni}. Their perpetual cycles
are an artefact of two assumptions no field population satisfies: prey that grow
without limit, and predators that keep eating faster as prey get denser. Put the
prey on logistic growth and let each predator saturate, as Holling's type-II
response does~\cite{holling1959components,holling1965functional}, and the neutral
cycles give way to a Hopf bifurcation, that is, to a boom-and-bust cycle of a
definite amplitude set by the
parameters~\cite{freedman1980deterministic,may2001stability}. Once the two
species are also allowed to disperse, the question becomes how that temporal
cycle and the Turing mechanism act on each
other~\cite{yi2009bifurcation,li2013hopf,murray2003mathematical,%
okubo2001diffusion,miao2022hopf}. A separate line of work makes the intake rate
depend on predator density as well, either because foragers get in each other's
way~\cite{beddington1975mutual,deangelis1975model} or because they move up prey
gradients~\cite{luo2021global}.

A model meant for augmentative biological control needs one more ingredient.
Released natural enemies are not left to live on the pest alone: they are given
pollen, factitious prey~\cite{toab218} or artificial
diet~\cite{verma2026t,two_aniket} so that they survive the gaps between
outbreaks. Srinivasu, Prasad and Venkatesulu~\cite{srinivasu2007biological} wrote
this into the equations by letting the predator feed on a supplement that is
spread evenly over the habitat and never runs out, of quantity $\xi$ and quality
$1/\alpha$. Eating the supplement uses up handling time, so $\xi$ enters the
denominator of the type-II response and weakens predation on the pest; the
supplement also feeds the predator, so it enters the numerator of the growth
term. The model has since been worked over thoroughly: when to deliver the
food~\cite{srinivasu2010time}, how much to deliver~\cite{srinivasu2011role}, what
happens when the predators interfere with one another~\cite{prasad2013dynamics},
what constant harvesting does~\cite{sen2015global}, what an Allee effect in the
prey does~\cite{gurubilli2017global}, and what changes under a type-III
response~\cite{srinivasu2018additional}. The recurring conclusion is that
high-quality food can push the pest to extinction, while poor-quality food tends
to shelter it.

There is a flaw in this framework that was noticed only recently. Parshad,
Wickramasooriya, Antwi-Fordjour and Banerjee~\cite{parshad2023additional} showed
that once the supplement is plentiful enough the predator density reaches
infinity in finite time. The reason is easy to state: the supplement is never
depleted, so each predator keeps growing at a rate that does not fall off as the
population rises, and nothing in the model stops it. The fix they propose is
biological rather than cosmetic. Predators compete among themselves for shelter,
oviposition sites and mates, and the competition is fiercest exactly where the
supplement has drawn them together, so a loss term proportional to $\xi y^{2}$
belongs in the equation. It restores boundedness, and it also changes which
bifurcations the model has rather than simply capping the dynamics. The
coefficient $c$ then becomes a quantity with a field meaning: how much the
released enemies get in each other's way.

What nobody has asked is what this term does once the populations occupy space,
and the question is not merely technical. Supplemented systems are known to
behave differently when dispersal is included; Ghorai and
Poria~\cite{ghorai2016pattern}, for instance, found that additional food can
either create or remove spatiotemporal chaos depending on its quality. Predator
self-limitation damps the density of exactly the species that plays the inhibitor
in the Turing mechanism, so it should act on pattern formation directly. Whether
it suppresses patterns, encourages them, or shifts the balance between fixed
patterns and oscillations is not something linear stability analysis can settle
on its own.

Settling it requires the codimension-two theory. Where a Turing curve and a Hopf
curve cross in a plane of two parameters, a zero eigenvalue at a finite
wavenumber sits alongside a pair of purely imaginary eigenvalues at zero
wavenumber, and a three-dimensional normal form decides what is actually
observed: a uniform oscillation, a stationary pattern, or a pattern that
oscillates. The recipe for computing its coefficients in reaction--diffusion
systems is due mainly to Song, Zhang and Peng~\cite{song2016turing} and to Jiang,
An and Shi~\cite{jiang2020formulation}. It has been applied to ratio-dependent
Holling--Tanner kinetics~\cite{an2018turing}, the Crowley--Martin
response~\cite{cao2018turing}, the Lengyel--Epstein system~\cite{chen2019turing},
Schnakenberg kinetics with a delay in gene expression~\cite{jiang2019turing}, the
Sel'kov--Schnakenberg family~\cite{li2023turing}, chemotaxis with a fear
effect~\cite{dai2021turing,lv2024turing}, schooling prey with Smith
growth~\cite{fu2025turing}, harvesting together with an Allee
effect~\cite{chen2023turing}, hunting cooperation and anti-predator
behavior~\cite{li2025spatiotemporal}, and a modified Leslie--Gower system with a
nonmonotonic response~\cite{li2025turing}; the attractors that such points
organise have been catalogued in
detail~\cite{an2019spatiotemporal,song2016turing}. We are not aware of any such
analysis for a model that carries additional food and predator self-limitation at
the same time, even though that is the combination which makes the
additional-food framework well posed in the first place.

This paper fills that gap. We take the kinetics of Parshad et
al.~\cite{parshad2023additional}, place them in a bounded one-dimensional habitat
with no flux across the boundaries, and use the prey diffusivity $d_{1}$ and the
competition strength $c$ as the two parameters that are varied. We (i) locate the
Hopf bifurcation of the kinetics in closed form and show, using Perko's Lyapunov
number, that it is supercritical across the whole admissible parameter window, so
the cycle that appears is a stable one; (ii) give necessary and sufficient
conditions for diffusion-driven instability and derive the normal form of the
steady-state bifurcation it produces, fixing the convention that decides whether
the branch is super- or subcritical; (iii) prove that the spatially uniform Hopf
bifurcation is supercritical, by computing the first Lyapunov coefficient in
Hassard's normalisation and checking it against Kuznetsov's invariant formula;
(iv) locate the Turing--Hopf point, verify that both curves cross transversally,
and compute the three-dimensional normal form on the center manifold; and (v)
test every prediction against direct numerical solution of the partial
differential equations. Sections~\ref{sec:model}--\ref{sec:hopf-ode} set up the
model and treat its kinetics, Sections~\ref{sec:rd}--\ref{sec:th} treat the
spatial problem, and Sections~\ref{sec:discussion}--\ref{sec:conclusion} read the
results back into the biology.

%% =====================================================================
\section{The model formulation}\label{sec:model}
%% =====================================================================

%% NOTE TO AUTHORS: in the Chaos source the Lotka--Volterra system, the logistic
%% law and the unmodified Holling type-II response p(x)=x/(1+x) were commented
%% out, so this section now opens directly on the additional-food modification.
%% A referee may find the jump abrupt, since x, y and the type-II form are used
%% below before being introduced. Consider restoring two or three sentences here.

In accordance with Srinivasu \emph{et al.}~\cite{srinivasu2007biological}, the
concept of additional food was introduced to improve the efficiency of
biological control by sustaining predators through periods of prey scarcity.
Supplementary food reduces starvation and enhances predator reproduction and
persistence. Srinivasu \emph{et al.}\ modified the Holling type-II response by
introducing the quantity and the quality of the additional food, so that the
predation term becomes
\begin{equation}
\frac{xy}{1+\alpha\xi+x},
\end{equation}
where $\xi$ is the quantity of additional food supplied to the predators and
$1/\alpha$ represents its quality. The presence of $\alpha\xi$ in the
denominator reflects the fact that predators divide their handling effort
between prey and supplement, so that the effective predation pressure on the
prey depends on both prey density and food supplementation.

Because predators draw energy from both sources, their reproduction term is
modified accordingly,
\begin{equation}
\frac{\beta(x+\xi)y}{1+\alpha\xi+x}
=\frac{\beta xy}{1+\alpha\xi+x}+\frac{\beta\xi y}{1+\alpha\xi+x},
\end{equation}
$\beta$ denoting the conversion efficiency. Predator biomass therefore increases
through consumption of prey as well as through utilization of the supplement,
allowing the predator to persist when the prey is scarce.

The formulation above still assumes that predator mortality is proportional to
predator density alone. In practice predators experience intense intraspecific
competition for space, shelter, mating opportunities and resources; as predator
density rises this competition reduces the per-capita growth rate. Parshad
\emph{et al.}~\cite{parshad2023additional} showed that this is not a cosmetic
refinement: without a density-dependent brake the additional-food model admits
solutions that blow up in finite time, because a non-depletable supplement keeps
the predator's per-capita growth rate bounded away from zero however large the
population becomes. Since the crowding is itself driven by the aggregating effect
of the supplement, the self-limitation term is taken proportional to the amount
of additional food,
\begin{equation}
c\,\xi y^{2},
\end{equation}
where $c>0$ measures the strength of predator intraspecific competition. Its
inclusion prevents unbounded predator growth and yields biologically meaningful
dynamics.

Collecting these ingredients, the predator--prey system with logistic prey
growth, a Holling type-II response modified by additional food, predator growth
supported by both prey and supplement, and density-dependent predator
competition, is
\begin{equation}\label{eq:ode}
\left\{
\begin{aligned}
\frac{\dd x}{\dd t}&=x\Bigl(1-\frac{x}{\gamma}\Bigr)-\frac{xy}{1+\alpha\xi+x},\\[2pt]
\frac{\dd y}{\dd t}&=\frac{\beta(x+\xi)y}{1+\alpha\xi+x}-\delta y-c\xi y^{2},\\[2pt]
x(0)&>0,\qquad y(0)>0 .
\end{aligned}\right.
\end{equation}

\begin{remark}[Scaling]\label{rem:scaling}
System \eqref{eq:ode} is written in dimensionless variables. Starting from the
dimensional model, time is rescaled by the prey intrinsic growth rate and the
prey density by the half-saturation constant of the functional response; this
normalizes the intrinsic growth rate to unity, which is why the parameter for
intrinsic growth rate $r$ no longer appears in \eqref{eq:ode}. The remaining
parameters $\gamma,\alpha,\beta,\delta,\xi,c$ are the corresponding
dimensionless groups, and all of them are positive.
\end{remark}

Here $x(t)$ and $y(t)$ are the prey and predator densities, $\gamma$ is the prey
carrying capacity, $\beta$ the conversion efficiency, $\delta$ the natural
predator mortality, $\xi$ the quantity of additional food, $1/\alpha$ its quality
and $c$ the coefficient of predator intraspecific competition. The model extends
the classical framework by incorporating additional food and predator
self-regulation simultaneously, and thereby permits a detailed study of how food
supplementation and predator competition affect persistence, stability and
bifurcation behavior.

It is convenient throughout to write \eqref{eq:ode} in the factored form
\begin{equation}\label{eq:factored}
\frac{\dd x}{\dd t}=p(x)\bigl(F(x)-y\bigr),\qquad
\frac{\dd y}{\dd t}=y\bigl(h(x)-c\xi y\bigr),
\end{equation}
with
\begin{equation}\label{eq:pFh}
p(x)=\frac{x}{1+\alpha\xi+x},\quad
F(x)=\Bigl(1-\frac{x}{\gamma}\Bigr)(1+\alpha\xi+x),\quad
h(x)=\frac{\beta(x+\xi)}{1+\alpha\xi+x}-\delta .
\end{equation}
The function $F$ is the prey nullcline written as a graph over $x$; it is a
downward parabola with $F(0)=1+\alpha\xi$, $F(\gamma)=0$ and a maximum at
$x=(\gamma-1-\alpha\xi)/2$. The function $h$ is the predator's net per-capita
growth rate at zero predator density; a short computation gives
\begin{equation}
h'(x)=\frac{\beta\bigl[1+\xi(\alpha-1)\bigr]}{(1+\alpha\xi+x)^{2}},
\end{equation}
so that $h$ is strictly increasing whenever $1+\xi(\alpha-1)>0$, which holds in
particular for all $\alpha\ge 1$. This monotonicity is used repeatedly below.

%% =====================================================================
\section{Equilibria and their stability}\label{sec:ode}
%% =====================================================================

System \eqref{eq:ode} possesses the trivial equilibrium $E_{0}=(0,0)$, the
prey-free equilibrium $E_{1}=(0,\bar y)$, the predator-free equilibrium
$E_{2}=(\gamma,0)$ and the coexistence equilibrium $\Es=(\xs,\ys)$. A brief
description of these states is given in Ref.~\cite{parshad2023additional}. For
completeness, and because the location of $\Es$ is needed at every subsequent
step, we record them explicitly.

\begin{proposition}\label{prop:equilibria}
Let all parameters be positive and let $p,F,h$ be as in \eqref{eq:pFh}.
\begin{enumerate}
\item $E_{0}=(0,0)$ always exists; its eigenvalues are $1$ and $h(0)$, so
$E_{0}$ is always unstable.
\item The prey-free equilibrium $E_{1}=(0,\bar y)$ with
$\bar y=h(0)/(c\xi)$ exists in the open positive quadrant if and only if
$h(0)=\beta\xi/(1+\alpha\xi)-\delta>0$, i.e.\ if and only if the supplement
alone is able to sustain the predator.
\item $E_{2}=(\gamma,0)$ always exists; it is a saddle if $h(\gamma)>0$ and
locally asymptotically stable if $h(\gamma)<0$.
\item A coexistence equilibrium $\Es=(\xs,\ys)$ is a solution of
\begin{equation}\label{eq:interior}
h(\xs)=c\xi F(\xs),\qquad \ys=F(\xs),\qquad \xs\in(0,\gamma).
\end{equation}
If $h(0)<0<h(\gamma)$ then \eqref{eq:interior} has at least one root in
$(0,\gamma)$, because $G(x):=h(x)-c\xi F(x)$ satisfies
$G(0)=h(0)-c\xi(1+\alpha\xi)<0$ and $G(\gamma)=h(\gamma)>0$.
\end{enumerate}
\end{proposition}

For the parameter set used throughout the numerical part of this paper,
\begin{equation}\label{eq:params}
\gamma=11,\quad \alpha=1,\quad \delta=1,\quad \beta=3.9,\quad \xi=0.3,
\end{equation}
one has $h(0)=-0.1<0$ and $h(\gamma)=2.5829268>0$, so that $E_{1}$ does not
exist, $E_{2}$ is a saddle, and a coexistence state is guaranteed. Numerical
continuation of $G$ on $(0,\gamma)$ shows that for every value of $c$ considered
below the root of \eqref{eq:interior} is unique.

%% =====================================================================
\section{Hopf bifurcation of the kinetics}\label{sec:hopf-ode}
%% =====================================================================

In the notation of \eqref{eq:factored}--\eqref{eq:pFh} the Jacobian of
\eqref{eq:ode} is
\begin{equation}
J=\begin{pmatrix}
p'(x)\bigl(F(x)-y\bigr)+F'(x)p(x) & -p(x)\\[2pt]
h'(x)\,y & h(x)-2c\xi y
\end{pmatrix},
\end{equation}
and at the coexistence equilibrium, where $\ys=F(\xs)$ and
$h(\xs)=c\xi F(\xs)$, it reduces to
\begin{equation}\label{eq:Jstar}
J^{*}=\begin{pmatrix}
F'(\xs)p(\xs) & -p(\xs)\\[2pt]
h'(\xs)F(\xs) & -c\xi F(\xs)
\end{pmatrix}.
\end{equation}
Hence
\begin{align}
\operatorname{tr}(J^{*})&=F'(\xs)p(\xs)-c\xi F(\xs),\label{eq:tr}\\
\det(J^{*})&=p(\xs)F(\xs)\bigl(h'(\xs)-c\xi F'(\xs)\bigr).\label{eq:det}
\end{align}
For \eqref{eq:factored} to undergo a Hopf bifurcation at $\Es$ the Jacobian
$J^{*}$ must possess a pair of purely imaginary eigenvalues
$\pm \ii\,\omega_{0}$, that is
\begin{equation}\label{eq:hopfcond}
\operatorname{tr}(J^{*})=0,\qquad \det(J^{*})=\omega_{0}^{2}>0 .
\end{equation}

\begin{remark}\label{rem:hopfsingle}
Combining the two conditions in \eqref{eq:hopfcond} with the equilibrium
relation $h(\xs)=c\xi F(\xs)$ eliminates $c$ and leaves the single scalar
equation
\begin{equation}\label{eq:hopfscalar}
h(\xs)=F'(\xs)\,p(\xs),
\end{equation}
which determines $\xs$ independently of $c$; the critical competition strength
then follows from $c_{0}=h(\xs)/\bigl(\xi F(\xs)\bigr)$. For the parameter set
\eqref{eq:params}, equation \eqref{eq:hopfscalar} has the unique root
\begin{equation}\label{eq:Estar}
\xs=0.0640448876,\qquad \ys=F(\xs)=1.3561030602,
\end{equation}
when
\begin{equation}\label{eq:c0}
c_{0}=0.1004263750,\qquad
\omega_{0}=\sqrt{\det J^{*}}=0.3630319203 ,
\end{equation}
with $\det J^{*}=0.1317921752$.
\end{remark}

\subsection{Direction of the Hopf bifurcation and stability of the
bifurcating periodic solution}\label{sec:hopfdir}

Under the Hopf condition $\operatorname{tr}(J^{*})=0$, $\det(J^{*})>0$, one has
\begin{equation}\label{eq:trzero}
F'(\xs)p(\xs)=c\xi F(\xs).
\end{equation}
Shifting coordinates by $X_{1}=x-\xs$, $Y_{1}=y-F(\xs)$ moves $\Es$ to the
origin. Expanding \eqref{eq:factored} in a Taylor series about $\Es$ gives
\begin{equation}\label{eq:taylor}
\left\{
\begin{aligned}
\dot X_{1}&=p(\xs)F'(\xs)X_{1}-p(\xs)Y_{1}
+\Bigl(p'(\xs)F'(\xs)+\tfrac{1}{2}p(\xs)F''(\xs)\Bigr)X_{1}^{2}
-p'(\xs)X_{1}Y_{1}\\
&\quad
+\Bigl(\tfrac{1}{6}F'''(\xs)p(\xs)+\tfrac{1}{2}F''(\xs)p'(\xs)
+\tfrac{1}{2}F'(\xs)p''(\xs)\Bigr)X_{1}^{3}
-\tfrac{1}{2}p''(\xs)X_{1}^{2}Y_{1}+O(|X_{1},Y_{1}|^{4}),\\[4pt]
\dot Y_{1}&=F(\xs)h'(\xs)X_{1}-c\xi F(\xs)Y_{1}
+\tfrac{1}{2}F(\xs)h''(\xs)X_{1}^{2}+h'(\xs)X_{1}Y_{1}-c\xi Y_{1}^{2}\\
&\quad+\tfrac{1}{6}F(\xs)h'''(\xs)X_{1}^{3}
+\tfrac{1}{2}h''(\xs)X_{1}^{2}Y_{1}+O(|X_{1},Y_{1}|^{4}).
\end{aligned}\right.
\end{equation}
Using \eqref{eq:trzero}, system \eqref{eq:taylor} takes the form
\begin{equation}\label{eq:abform}
\left\{
\begin{aligned}
\dot X_{1}&=a_{10}X_{1}+a_{01}Y_{1}+a_{20}X_{1}^{2}+a_{11}X_{1}Y_{1}
+a_{30}X_{1}^{3}+a_{21}X_{1}^{2}Y_{1}+O(4),\\
\dot Y_{1}&=b_{10}X_{1}-a_{10}Y_{1}+b_{20}X_{1}^{2}+b_{02}Y_{1}^{2}
+b_{11}X_{1}Y_{1}+b_{30}X_{1}^{3}+b_{21}X_{1}^{2}Y_{1}+O(4),
\end{aligned}\right.
\end{equation}
with
\begin{align}
&a_{10}=p(\xs)F'(\xs),\quad a_{01}=-p(\xs),\notag\\
&a_{20}=p'(\xs)F'(\xs)+\tfrac12 p(\xs)F''(\xs),\quad a_{11}=-p'(\xs),\notag\\
&a_{30}=\tfrac16 F'''(\xs)p(\xs)+\tfrac12 F''(\xs)p'(\xs)+\tfrac12 F'(\xs)p''(\xs),\notag\\
&a_{21}=-\tfrac12 p''(\xs),\quad b_{10}=F(\xs)h'(\xs),\quad b_{20}=\tfrac12 F(\xs)h''(\xs),\notag\\
&b_{02}=-c\xi,\quad b_{11}=h'(\xs),\quad
b_{30}=\tfrac16 F(\xs)h'''(\xs),\quad b_{21}=\tfrac12 h''(\xs).\label{eq:abcoef}
\end{align}
To bring \eqref{eq:abform} into canonical form we apply
\begin{equation}
X_{1}=X_{2},\qquad
Y_{1}=-\frac{a_{10}}{a_{01}}X_{2}-\frac{\omega_{0}}{a_{01}}Y_{2},
\end{equation}
which yields
\begin{equation}\label{eq:canon}
\left\{
\begin{aligned}
\dot X_{2}&=-\omega_{0}Y_{2}+A_{20}X_{2}^{2}+A_{11}X_{2}Y_{2}
+A_{30}X_{2}^{3}+A_{21}X_{2}^{2}Y_{2}+O(4),\\
\dot Y_{2}&=\omega_{0}X_{2}+B_{20}X_{2}^{2}+B_{11}X_{2}Y_{2}+B_{02}Y_{2}^{2}
+B_{30}X_{2}^{3}+B_{21}X_{2}^{2}Y_{2}+O(4),
\end{aligned}\right.
\end{equation}
where
\begin{align}
&A_{20}=a_{20}-\frac{a_{11}a_{10}}{a_{01}},\quad
A_{11}=-\frac{a_{11}\omega_{0}}{a_{01}},\quad
A_{30}=a_{30}-\frac{a_{21}a_{10}}{a_{01}},\quad
A_{21}=-\frac{a_{21}\omega_{0}}{a_{01}},\notag\\
&B_{20}=-\frac{1}{\omega_{0}}\Bigl(b_{20}a_{01}-b_{11}a_{10}+a_{20}a_{10}
+\frac{b_{02}a_{10}^{2}}{a_{01}}-\frac{a_{11}a_{10}^{2}}{a_{01}}\Bigr),\notag\\
&B_{11}=\frac{a_{10}a_{11}}{a_{01}}+b_{11}-\frac{2b_{02}a_{10}}{a_{01}},\quad
B_{02}=-\frac{b_{02}\omega_{0}}{a_{01}},\notag\\
&B_{30}=\frac{a_{10}}{\omega_{0}}\Bigl(\frac{a_{21}a_{10}}{a_{01}}-a_{30}
-\frac{b_{30}a_{01}}{a_{10}}+b_{21}\Bigr),\quad
B_{21}=\frac{a_{21}a_{10}}{a_{01}}+b_{21}.
\end{align}
Following Perko~\cite[p.~353]{perko2013differential} (see also
Refs.~\cite{guckenheimer1983nonlinear,kuznetsov2004elements} for equivalent
formulations), the Lyapunov number of \eqref{eq:canon} is
\begin{equation}\label{eq:liapunov}
\sigma=\frac{3\pi}{2\,\omega_{0}^{2}}
\Bigl[3\omega_{0}A_{30}+\omega_{0}B_{21}
+A_{11}A_{20}-B_{11}(B_{20}+B_{02})-2A_{20}B_{20}\Bigr].
\end{equation}
The direction of the Hopf bifurcation is determined by the sign of $\sigma$:
$\sigma<0$ gives a supercritical bifurcation with stable small-amplitude
periodic solutions, while $\sigma>0$ gives a subcritical bifurcation with
unstable cycles.

\subsection{Numerical verification}\label{sec:hopfnum}

For the parameter set \eqref{eq:params} the coefficients \eqref{eq:abcoef}
evaluated at the Hopf point \eqref{eq:Estar}--\eqref{eq:c0} are
\begin{equation}
\begin{aligned}
&a_{10}=0.0408565543,\quad a_{01}=-0.0469521848,\quad a_{20}=0.6037153524,\\
&a_{11}=-0.6986924139,\quad a_{30}=-0.5092386985,\quad a_{21}=0.5122209835,\\
&b_{10}=2.8424967620,\quad b_{20}=-2.0838733299,\quad b_{02}=-0.0301279125,\\
&b_{11}=2.0960772418,\quad b_{30}=1.5277160956,\quad b_{21}=-1.5366629506,
\end{aligned}
\end{equation}
which give
\begin{equation}
\begin{aligned}
&A_{20}=-0.0042683804,\ A_{11}=-5.4022544415,\ A_{30}=-0.0635174922,\\
&A_{21}=3.9604667637,\ B_{20}=-0.0360867362,\ B_{11}=2.6516279430,\\
&B_{02}=-0.2329474971,\ B_{30}=0.0317932400,\ B_{21}=-1.9823841570,
\end{aligned}
\end{equation}
and therefore
\begin{equation}\label{eq:sigmaval}
\sigma=-1.8849193<0 .
\end{equation}
The bifurcation is thus supercritical. It is worth noting that $c$ is not a free
parameter here: by Remark~\ref{rem:hopfsingle} the Hopf condition fixes $\xs$ and
hence $c_{0}$, so $\sigma$ is defined at one point of the $c$-axis only. To test
robustness we therefore followed the Hopf curve in the $(\beta,\xi)$-plane,
recomputing $\xs$, $c_{0}$ and $\sigma$ at each point. For
$\beta\in\{3.5,3.9,4.3\}$ and $\xi\in\{0.1,0.2,0.3\}$, the range over which a
biologically admissible Hopf point exists for $\gamma=11$, $\alpha=\delta=1$, the
critical competition strength varies over almost three orders of magnitude, from
$c_{0}=1.593$ down to $c_{0}=0.0071$, while $\sigma$ stays between $-1.11$ and
$-6.47$. Supercriticality is therefore not an artefact of the particular
parameter set \eqref{eq:params}. Three further independent checks confirm the
conclusion at \eqref{eq:params}. First, the first Lyapunov coefficient computed
in Sec.~\ref{sec:hopf-pde} from Hassard's normalization gives
$\ell_{1}=-0.0726049860<0$, and one verifies the identity
$\sigma=3\pi\ell_{1}/\omega_{0}$ to ten digits. Second, Kuznetsov's invariant
formula~\cite{kuznetsov2004elements}
\begin{equation}
\ell_{1}^{\mathrm{inv}}=\frac{1}{2\omega_{0}}\,
\mathrm{Re}\Bigl[\ip{q^{*}}{C(q,q,\bar q)}
-2\ip{q^{*}}{B\bigl(q,A^{-1}B(q,\bar q)\bigr)}
+\ip{q^{*}}{B\bigl(\bar q,(2\ii\omega_{0}I-A)^{-1}B(q,q)\bigr)}\Bigr]
\end{equation}
returns $\ell_{1}^{\mathrm{inv}}=-0.1999961<0$. Third, direct integration of
\eqref{eq:ode} shows that at $c=c_{0}$ small perturbations of $\Es$ decay, that
for $c>c_{0}$ perturbations as large as $80\%$ of $\xs$ return to $\Es$ (so no
unstable cycle coexists with the stable equilibrium), and that stable limit
cycles are found only for $c<c_{0}$. All three tests are consistent with a
supercritical Hopf bifurcation; a representative cycle is shown in
Fig.~\ref{fig:1}.

\begin{figure}[ht!] 
\includegraphics[width=0.8\columnwidth]{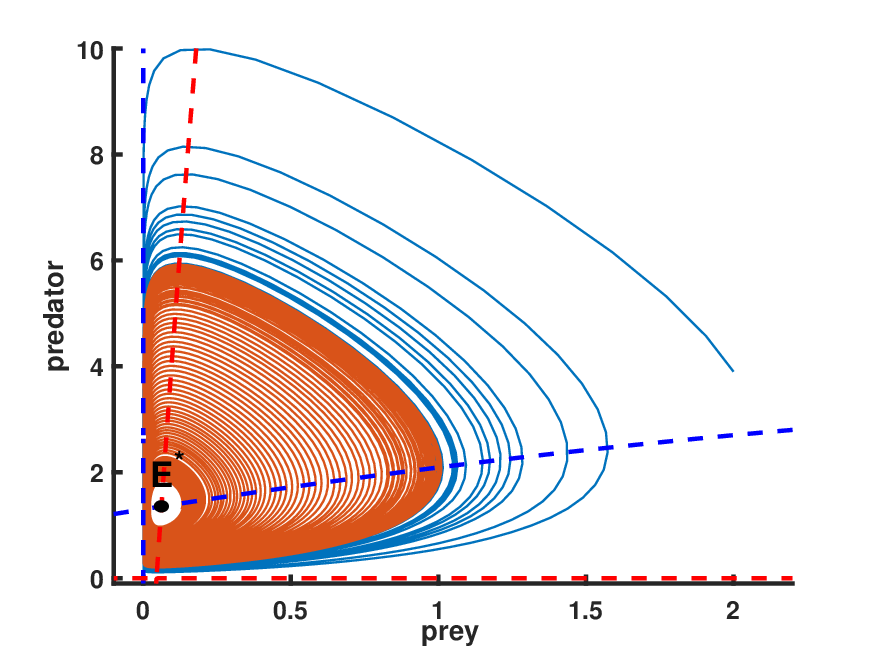} \caption{Phase-plane representation for $\gamma=11$, $\alpha=1$, $\delta=1$, $\beta=3.9$, $\xi=0.3$, and $c=0.088$. The coexistence equilibrium $\Es$ is unstable and is surrounded by a stable limit cycle.} \label{fig:1}
\alttext{Phase-plane plot of predator density against prey density. A trajectory starting near the coexistence equilibrium spirals outward and settles onto a closed orbit, so the equilibrium repels nearby states while the surrounding periodic orbit attracts them.}
\end{figure}

%% =====================================================================
\section{The reaction--diffusion system and Turing instability}\label{sec:rd}
%% =====================================================================

We now embed the kinetics \eqref{eq:factored} in a bounded one-dimensional
habitat and consider
\begin{equation}\label{eq:pde}
\left\{
\begin{aligned}
\frac{\partial x}{\partial t}&=d_{1}\Delta x+p(x)\bigl(F(x)-y\bigr),
&& u\in\Om,\ t>0,\\
\frac{\partial y}{\partial t}&=d_{2}\Delta y+y\bigl(h(x)-c\xi y\bigr),
&& u\in\Om,\ t>0,\\
\frac{\partial x}{\partial n}&=\frac{\partial y}{\partial n}=0,
&& u\in\partial\Om,\ t>0,\\
x(u,0)&=x_{0}(u)\ge 0,\quad y(u,0)=y_{0}(u)\ge 0, && u\in\Om,
\end{aligned}\right.
\end{equation}
where $x(u,t)$ and $y(u,t)$ are the prey and predator densities at the spatial
position $u$, $d_{1},d_{2}>0$ are the diffusion coefficients and $n$ is the
outward unit normal on $\partial\Om$. Throughout, $\Om=(0,l\pi)\subset\R$ with
$l>0$ and $\Delta=\partial^{2}/\partial u^{2}$; the spatial variable is $u$, and
the symbols $x$ and $y$ are reserved for the population densities.

We work in the real Sobolev space
\begin{equation}
\mathcal{X}=\Bigl\{(x,y)\in H^{2}(\Om)\times H^{2}(\Om):
\frac{\partial x}{\partial n}=\frac{\partial y}{\partial n}=0,\
u\in\partial\Om\Bigr\},
\end{equation}
and denote by $\mathcal{X}_{\mathbb{C}}$ its complexification,
\begin{equation}
\mathcal{X}_{\mathbb{C}}=\mathcal{X}\oplus \ii\,\mathcal{X}
=\{x_{1}+\ii x_{2}:x_{1},x_{2}\in\mathcal{X}\}.
\end{equation}

\begin{remark}[Well-posedness]\label{rem:wellposed}
The reaction terms in \eqref{eq:pde} are smooth on the closed positive quadrant
and both coordinate axes are invariant, so the non-negative quadrant is
positively invariant for the kinetics; since the diffusion matrix is diagonal
with positive entries, the same holds for \eqref{eq:pde} by the maximum
principle~\cite{murray2003mathematical,okubo2001diffusion}. Moreover
$\dd x/\dd t\le x(1-x/\gamma)$ gives
$\limsup_{t\to\infty}\max_{\bar\Om}x\le\gamma$, and on the invariant region
$\{x\le\gamma\}$ one has $\dd y/\dd t\le y\bigl(h(\gamma)-c\xi y\bigr)$, so that
$\limsup_{t\to\infty}\max_{\bar\Om}y\le h(\gamma)/(c\xi)$. Solutions of
\eqref{eq:pde} with non-negative initial data therefore exist globally and are
uniformly bounded. This is precisely the role played by the self-limitation
term: without it the second comparison fails and solutions may blow up in finite
time~\cite{parshad2023additional}.
\end{remark}

Translating $x=\bar x+\xs$, $y=\bar y+\ys$, dropping the bars and linearising
\eqref{eq:pde} at $\Es$ gives
\begin{equation}\label{eq:lin}
\dot U(t)=\begin{pmatrix} d_{1}&0\\ 0&d_{2}\end{pmatrix}\Delta U(t)
+\begin{pmatrix} a_{1}&a_{2}\\ b_{1}&cb_{2}\end{pmatrix}U(t),
\end{equation}
where $U=(x,y)^{T}$ and
\begin{equation}\label{eq:abdef}
a_{1}=F'(\xs)p(\xs),\quad a_{2}=-p(\xs),\quad
b_{1}=h'(\xs)\ys,\quad b_{2}=-\xi \ys .
\end{equation}
The operator $\varphi\mapsto-\Delta\varphi$ on $(0,l\pi)$ with
$\varphi'(0)=\varphi'(l\pi)=0$ has eigenvalues $z_{k}=k^{2}/l^{2}$ with
eigenfunctions $\varphi_{k}(u)=\cos(ku/l)$, $k\in\N_{0}=\N\cup\{0\}$.
Consequently the characteristic equation of \eqref{eq:lin} restricted to the
$k$-th mode is
\begin{equation}\label{eq:char}
\lambda^{2}-T_{k}\lambda+D_{k}=0,
\end{equation}
with
\begin{align}
T_{k}&=-(d_{1}+d_{2})z_{k}+(a_{1}+cb_{2}),\label{eq:Tk}\\
D_{k}&=d_{1}d_{2}z_{k}^{2}-(a_{1}d_{2}+cb_{2}d_{1})z_{k}+(ca_{1}b_{2}-a_{2}b_{1}),
\label{eq:Dk}
\end{align}
and roots $\lambda_{1,2}=\bigl(T_{k}\pm\sqrt{T_{k}^{2}-4D_{k}}\bigr)/2$.
Introduce the conditions
\begin{equation}
(R_{1}):\ a_{1}+cb_{2}<0,\qquad
(R_{2}):\ ca_{1}b_{2}-a_{2}b_{1}>0 .
\end{equation}
If $(R_{1})$ and $(R_{2})$ hold then $T_{0}<0$ and $D_{0}>0$, whence:

\begin{theorem}\label{thm:odestable}
If $(R_{1})$ and $(R_{2})$ hold, the coexistence equilibrium $\Es$ of the
kinetics \eqref{eq:factored} is locally asymptotically stable.
\end{theorem}

\subsection{Turing instability}\label{sec:turing}

Assume that $(R_{1})$ and $(R_{2})$ hold.

\begin{theorem}\label{thm:turing}
The coexistence equilibrium $\Es$ of \eqref{eq:pde} is locally
asymptotically stable if either
\begin{equation}
(R_{3}):\ a_{1}d_{2}+cb_{2}d_{1}<0,
\end{equation}
or
\begin{equation}
(R_{4}):\ a_{1}d_{2}+cb_{2}d_{1}>0 \ \text{ and } \ S<0,
\end{equation}
and is unstable if
\begin{equation}
(R_{5}):\ a_{1}d_{2}+cb_{2}d_{1}>0 \ \text{ and } \ S>0,
\end{equation}
where $S=(a_{1}d_{2}+cb_{2}d_{1})^{2}-4d_{1}d_{2}(ca_{1}b_{2}-a_{2}b_{1})$.
Under $(R_{5})$ the equilibrium is stable for the kinetics and unstable for
\eqref{eq:pde}, so that the Turing bifurcation surface of $\Es$ is given by
$a_{1}d_{2}+cb_{2}d_{1}>0$ together with $S=0$.
\end{theorem}

\begin{proof}
The argument is the standard dispersion-relation
analysis~\cite{murray2003mathematical,yi2009bifurcation}. From \eqref{eq:Tk},
$T_{k+1}<T_{k}$ for all $k\ge0$, and $(R_{1})$ gives $T_{k}<0$ for all $k\ge0$.
If $(R_{3})$ holds then, by $(R_{2})$, $D_{k}>0$ for all $k\ge 0$ and every root
of \eqref{eq:char} has negative real part. If $(R_{4})$ holds then the quadratic
\eqref{eq:Dk} in $z_{k}$ has no real root, so again $D_{k}>0$ for all $k$. If
$(R_{5})$ holds there exists $k^{*}\in\N$ with $D_{k^{*}}<0$, so \eqref{eq:char}
has a positive real root and $\Es$ loses stability: this is the Turing
instability~\cite{turing1952chemical,segel1972dissipative}.
\end{proof}

\subsection{Normal form of the Turing bifurcation}\label{sec:turingnf}

Assume $(R_{1})$, $(R_{2})$, $a_{1}d_{2}+cb_{2}d_{1}>0$ and $S=0$. There is then
$k^{*}\in\N$ for which the eigenvalue associated with the critical mode is zero
and simple for the operator
\begin{equation}
L=\begin{pmatrix}
F'(\xs)p(\xs)+d_{1}\Delta & -p(\xs)\\
h'(\xs)F(\xs) & -c\xi F(\xs)+d_{2}\Delta
\end{pmatrix}.
\end{equation}
Let $q$ be the eigenvector of $L$ for the zero eigenvalue and $q^{*}$ the
corresponding eigenvector of the adjoint $L^{*}$, normalized by
$\ip{q^{*}}{q}=1$; here
\begin{equation}
L^{*}=\begin{pmatrix}
F'(\xs)p(\xs)+d_{1}\Delta & h'(\xs)F(\xs)\\
-p(\xs) & -c\xi F(\xs)+d_{2}\Delta
\end{pmatrix}.
\end{equation}
Translating $X_{1}=x-\xs$, $Y_{1}=y-\ys$ turns \eqref{eq:pde} into
\begin{equation}\label{eq:XLX}
\frac{\partial X}{\partial t}=LX+\tfrac12 R_{2}(X,X)+\tfrac16 R_{3}(X,X,X)
+O(\|X\|^{4}),\qquad u\in\Om,\ t>0,
\end{equation}
where $X=(X_{1},Y_{1})\in H^{2}(\Om)\times H^{2}(\Om)$ and, for
$U=(U_{1},U_{2})^{T}$, $V=(V_{1},V_{2})^{T}$, $W=(W_{1},W_{2})^{T}$,
\begin{align}
R_{2}(U,V)&=\begin{pmatrix}
2a_{20}U_{1}V_{1}+a_{11}(U_{1}V_{2}+U_{2}V_{1})\\
2b_{20}U_{1}V_{1}+2b_{02}U_{2}V_{2}+b_{11}(U_{1}V_{2}+U_{2}V_{1})
\end{pmatrix},\\
R_{3}(U,V,W)&=\begin{pmatrix}
6a_{30}U_{1}V_{1}W_{1}+2a_{21}(U_{1}V_{1}W_{2}+U_{1}V_{2}W_{1}+U_{2}V_{1}W_{1})\\
6b_{30}U_{1}V_{1}W_{1}+2b_{21}(U_{1}V_{1}W_{2}+U_{1}V_{2}W_{1}+U_{2}V_{1}W_{1})
\end{pmatrix},
\end{align}
the coefficients $a_{ij}$, $b_{ij}$ being those of \eqref{eq:abcoef}.

Near the bifurcation the solution decomposes as $X=zq+w$, where $zq$ is the
critical part and $w$ lies in the complementary stable subspace
$\mathcal{X}_{s}$~\cite{wu1996theory,song2016turing}. Substituting into
\eqref{eq:XLX},
\begin{equation}\label{eq:zw}
\left\{
\begin{aligned}
\dot z&=\tfrac12\ip{q^{*}}{R_{2}(zq+w,zq+w)}
+\tfrac16\ip{q^{*}}{R_{3}(zq+w,zq+w,zq+w)}+O(|z|^{4}),\\
\dot w&=Lw+\tfrac12\bigl[R_{2}(zq+w,zq+w)-q\ip{q^{*}}{R_{2}(zq+w,zq+w)}\bigr]\\
&\quad+\tfrac16\bigl[R_{3}(\cdot)-q\ip{q^{*}}{R_{3}(\cdot)}\bigr]+O(|w|^{4}).
\end{aligned}\right.
\end{equation}
Writing $w(z)=w_{20}z^{2}/2+O(|z|^{3})$ and comparing coefficients of $z^{2}$
gives $Lw_{20}=-\bigl[R_{2}(q,q)-q\ip{q^{*}}{R_{2}(q,q)}\bigr]$, so
\begin{equation}
w_{20}=-\bigl(L|_{\mathcal{X}_{s}}\bigr)^{-1}
\bigl[R_{2}(q,q)-\ip{q^{*}}{R_{2}(q,q)}q\bigr].
\end{equation}
Defining
\begin{equation}
g_{2}^{(k)}=\tfrac12\ip{q^{*}}{R_{2}(q,q)},\qquad
g_{3}^{(k)}=\tfrac12\ip{q^{*}}{R_{2}(q,w_{20})}
+\tfrac16\ip{q^{*}}{R_{3}(q,q,q)},
\end{equation}
the reduced equation is
\begin{equation}\label{eq:turingnf}
\dot z=g_{2}^{(k)}z^{2}+g_{3}^{(k)}z^{3}+O(|z|^{4}).
\end{equation}
For a critical mode with $k^{*}\ge1$ the quadratic coefficient vanishes
identically, $g_{2}^{(k^{*})}=0$, because the projection of
$R_{2}(q\cos(k^{*}u/l),q\cos(k^{*}u/l))$ onto $\cos(k^{*}u/l)$ involves
$\int_{0}^{l\pi}\cos^{3}(k^{*}u/l)\,\dd u=0$. The bifurcation is therefore always
a pitchfork and is governed by the cubic term,
\begin{equation}\label{eq:cubicnf}
\dot z=g_{3}^{(k^{*})}z^{3}+O(|z|^{4}).
\end{equation}

\begin{remark}[Criticality]\label{rem:criticality}
Following the analysis of Jiang~\cite{jiang2019schooling}, when the system
reaches the Turing bifurcation threshold at the positive constant equilibrium
$\Es$ a pitchfork bifurcation occurs. With $d_{1}$ as bifurcation parameter,
$D_{k^{*}}<0$ for $d_{1}$ slightly below $d_{1}^{k^{*}}$, so the amplitude
equation near the threshold reads
$\dot z=\mu(d_{1})z+g_{3}^{(k^{*})}z^{3}$ with $\mu>0$ for
$d_{1}<d_{1}^{k^{*}}$. The pitchfork is therefore \emph{supercritical}, with a
pair of stable spatially heterogeneous steady states, when $g_{3}^{(k^{*})}<0$,
and subcritical when $g_{3}^{(k^{*})}>0$.
\end{remark}

\subsection{Spatial patterns generated by the Turing bifurcation}
\label{sec:turingnum}

We take $\gamma=11$, $\alpha=1$, $\delta=1$, $\beta=3.9$, $\xi=0.3$, $d_{2}=1$,
$c=1.1$ and $\Om=(0,\pi)$, i.e.\ $l=1$ and $z_{k}=k^{2}$. The coexistence
equilibrium is
\begin{equation}
\Es=(\xs,\ys)=(0.3421406015,\,1.5910639680),
\end{equation}
and the entries \eqref{eq:abdef} of the linearisation are
\begin{equation}
a_{1}=0.1707662132,\quad a_{2}=-0.2083503698,\quad
b_{1}=2.3010792140,\quad b_{2}=-0.4773191903,
\end{equation}
so that $cb_{2}=-0.5250511093$, $\operatorname{tr}J^{*}=-0.3542848961<0$ and
\begin{equation}
D_{0}=\det J^{*}=0.3897697155>0 .
\end{equation}
The dispersion relation \eqref{eq:Dk} therefore reads
\begin{equation}\label{eq:dispersion11}
D_{k}=d_{1}z_{k}^{2}-\bigl(0.1707662132-0.5250511093\,d_{1}\bigr)z_{k}
+0.3897697155 ,
\end{equation}
and solving $D_{k}=0$ for $d_{1}$ mode by mode gives the Turing thresholds
\begin{equation}\label{eq:thresholds11}
d_{1}^{k=2}=0.0162039682,\quad
d_{1}^{k=3}=0.0133813946,\quad
d_{1}^{k=4}=0.0088596159,
\end{equation}
with $d_{1}^{k=1}<0$ and $d_{1}^{k}$ decreasing for $k\ge2$. The primary Turing
threshold is thus $d_{1}=d_{1}^{k=2}=0.0162039682$: the constant steady state is
stable for $d_{1}>d_{1}^{k=2}$ and loses stability to the mode $\cos(2u)$ as
$d_{1}$ decreases through this value.

At $d_{1}=0.013$ (Fig.~\ref{fig:2}) one finds $D_{2}=-0.0579925<0$,
$D_{3}=-0.0326952<0$ and $D_{k}>0$ for $k\ne2,3$, so two modes are linearly
unstable; the cubic coefficient of the primary mode, evaluated at its threshold
$d_{1}^{k=2}$, is $g_{3}^{(2)}=-0.0730288<0$, the pitchfork is supercritical, and
direct integration of \eqref{eq:pde} converges to a stationary pattern dominated
by $\cos(2u)$, in agreement with Remark~\ref{rem:criticality}. At $d_{1}=0.0087$
(Fig.~\ref{fig:3}) the unstable band widens to $k=2,3,4$, and the corresponding
coefficient at the threshold $d_{1}^{k=4}$ is $g_{3}^{(4)}=-0.2583011<0$. We
emphasize, however, that the $k=4$ branch, although it bifurcates supercritically
from the trivial state at $d_{1}=d_{1}^{k=4}$, is \emph{not} the attractor at
$d_{1}=0.0087$: numerical integration started from a pure $\cos(4u)$ perturbation
of $\Es$ drifts to the $\cos(2u)$-dominated state, as it must, since the $k=4$
branch inherits the instability of the background at that parameter value.
Figure~\ref{fig:3} should therefore be read as a transient or as a solution of
the mode-$4$ amplitude equation rather than as a stable attractor of the full
system.

\begin{figure}[ht!]
\centering
\begin{minipage}{0.48\linewidth}\figph{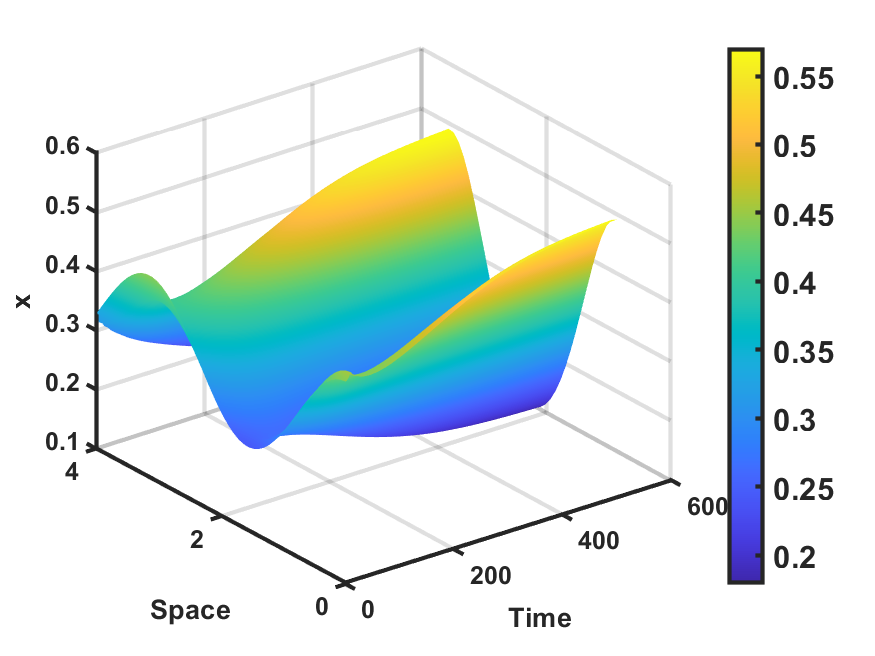}{45mm}\end{minipage}\hfill
\begin{minipage}{0.48\linewidth}\figph{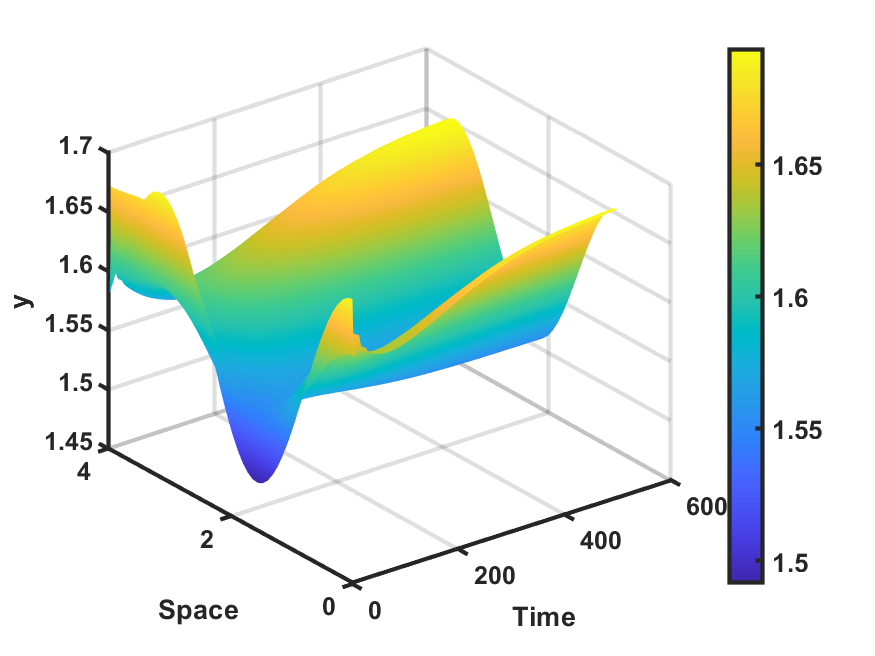}{45mm}\end{minipage}
\caption{\label{fig:2}%
$\gamma=11$, $\alpha=1$, $\delta=1$, $\beta=3.9$, $\xi=0.3$, $d_{1}=0.013$,
$d_{2}=1$, $c=1.1$, $\Om=(0,\pi)$. Initial data
$\bigl(x_{0}(u),y_{0}(u)\bigr)=(0.3421+0.1\cos 2u,\ 1.5911+0.1\cos 2u)$. The
constant state $\Es=(0.3421,1.5911)$ is unstable and \eqref{eq:pde} settles
on a stable spatially heterogeneous steady state with profile
$\propto\cos 2u$.}
\alttext{Two three-dimensional surface plots of prey and of predator density over space and time. Both fields develop a non-uniform profile that stops changing once transients decay, so the final state varies with position but not with time.}
\end{figure}

\begin{figure}[htb!]
\centering
\begin{minipage}{0.48\linewidth}\figph{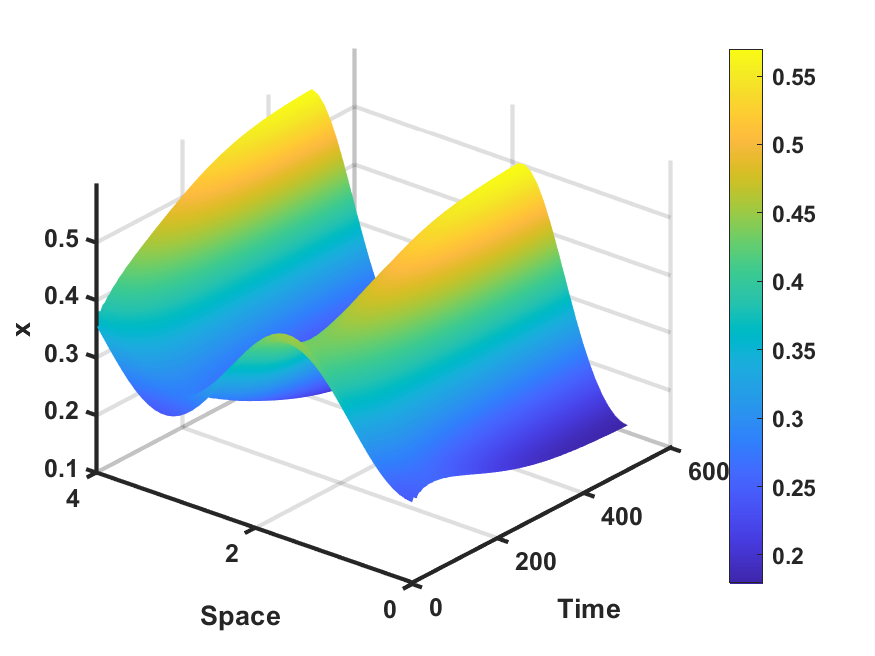}{45mm}\end{minipage}\hfill
\begin{minipage}{0.48\linewidth}\figph{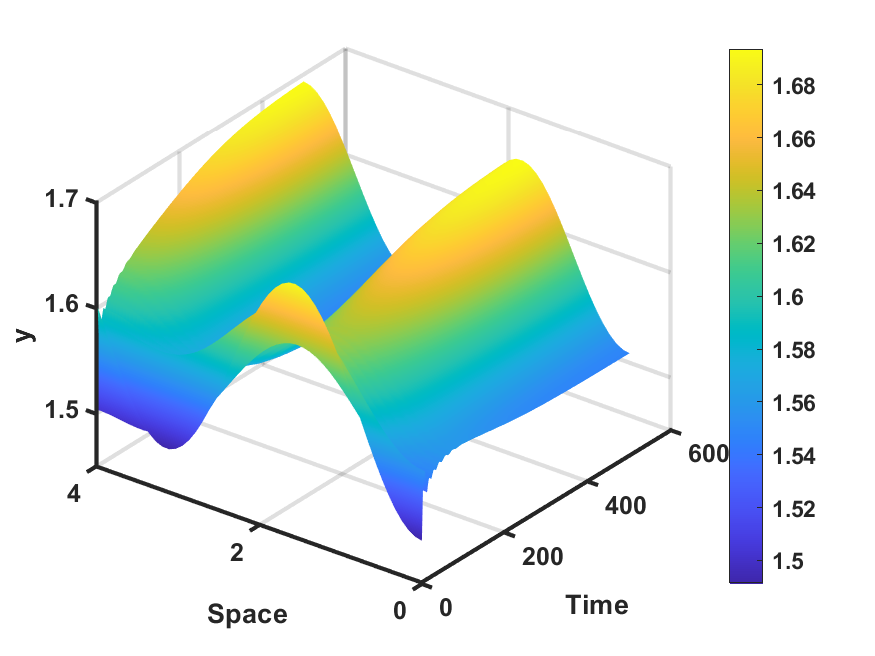}{45mm}\end{minipage}
\caption{\label{fig:3}%
Same parameters as Fig.~\ref{fig:2} but $d_{1}=0.0087$, with initial data
$\bigl(x_{0}(u),y_{0}(u)\bigr)=(0.3421-0.01\cos 4u,\ 1.5911-0.01\cos 4u)$,
showing the mode-$4$ heterogeneous state that bifurcates supercritically at
$d_{1}=d_{1}^{k=4}=0.0088596$. At this value of $d_{1}$ the modes $k=2$
and $k=3$ are also linearly unstable, so the mode-$4$ state is not the global
attractor.}
\alttext{Two three-dimensional surface plots of prey and of predator density over space and time. A higher-wavenumber profile grows from the initial perturbation and becomes stationary, with a broad plateau separated by a sharp front from a low-density region.}
\end{figure}

%% =====================================================================
\section{Spatially homogeneous Hopf bifurcation}\label{sec:hopf-pde}
%% =====================================================================

For a Hopf bifurcation of \eqref{eq:pde} the characteristic equation
\eqref{eq:char} must have purely imaginary roots, i.e.\ $T_{k}=0$ and $D_{k}>0$.
If $(R_{2})$ holds and
$S=(a_{1}d_{2}+cb_{2}d_{1})^{2}-4d_{1}d_{2}(ca_{1}b_{2}-a_{2}b_{1})<0$, then
$D_{k}>0$ for all $k\ge0$. Solving $T_{k}=0$ for $c$ gives
\begin{equation}
c=c_{k}=\frac{(d_{1}+d_{2})z_{k}-a_{1}}{b_{2}},\qquad z_{k}=\frac{k^{2}}{l^{2}} .
\end{equation}
At $c=c_{0}=-a_{1}/b_{2}$ one has $T_{0}=0$ and $T_{k}<0$ for all $k\ge1$, so
that the bifurcation is spatially homogeneous.

\begin{theorem}\label{thm:homhopf}
Assume $(R_{2})$ and $S<0$, so that $D_{k}>0$ for all $k\ge0$. Then
\eqref{eq:pde} undergoes a spatially homogeneous Hopf bifurcation at $\Es$
when $c=c_{0}=-a_{1}/b_{2}$, provided the transversality condition holds.
\end{theorem}

At the bifurcation the operator $L$ has the purely imaginary eigenvalues
$\pm \ii\omega_{0}$ of Sec.~\ref{sec:hopf-ode}. Put
\begin{equation}
\eta_{1}=\frac{-a_{1}+\ii\omega_{0}}{a_{2}},\qquad
\eta_{2}=\frac{-a_{1}-\ii\omega_{0}}{b_{1}},
\end{equation}
and define
\begin{equation}
q^{H}=\begin{pmatrix}1\\ \eta_{1}\end{pmatrix},\qquad
q^{H*}=\frac{a_{2}b_{1}}{\pi\bigl(a_{2}b_{1}+(a_{1}+\ii\omega_{0})^{2}\bigr)}
\begin{pmatrix}1\\ \eta_{2}\end{pmatrix}.
\end{equation}
Here the inner product on $\mathcal{X}_{\mathbb{C}}$ is the Hermitian product
\begin{equation}\label{eq:ipdef}
\ip{a}{b}=\int_{0}^{l\pi}\overline{a}^{\,T}b\,\dd r ,
\end{equation}
which is conjugate-linear in its first argument. With this normalization
$\ip{q^{H*}}{q^{H}}=1$, $Lq^{H}=\ii\omega_{0}q^{H}$ and
$L^{*}q^{H*}=-\ii\omega_{0}q^{H*}$.

Decompose $\mathcal{X}=\mathcal{X}^{c}\oplus\mathcal{X}^{s}$ with
$\mathcal{X}^{c}=\{zq^{H}+\bar z\bar q^{H}:z\in\mathbb{C}\}$ and
$\mathcal{X}^{s}=\{w\in\mathcal{X}:\ip{q^{H*}}{w}=0\}$~\cite{wu1996theory,li2013hopf},
so that every $X=(x,y)\in\mathcal{X}$ can be written
\begin{equation}\label{eq:decomp}
X=zq^{H}+\bar z\bar q^{H}+w .
\end{equation}
In the $(z,w)$ coordinates \eqref{eq:XLX} becomes
\begin{equation}\label{eq:zwH}
\left\{
\begin{aligned}
\dot z&=\ii\omega_{0}z+\ip{q^{H*}}{R_{2}(X)}+\ip{q^{H*}}{R_{3}(X)}+O(|z|^{4}),\\
\dot w&=Lw+R_{2}(X)-\ip{q^{H*}}{R_{2}(X)}q^{H}+R_{3}(X)
-\ip{q^{H*}}{R_{3}(X)}q^{H}+O(|w|^{4}).
\end{aligned}\right.
\end{equation}
Setting
$w(z,\bar z)=w_{20}z^{2}/2+w_{11}z\bar z+w_{02}\bar z^{2}/2+O(|z|^{3})$ and
expanding,
\begin{align}
R_{2}(X)&=R_{2}(q^{H},q^{H})z^{2}+2R_{2}(q^{H},\bar q^{H})z\bar z
+R_{2}(\bar q^{H},\bar q^{H})\bar z^{2}+O(|z|^{3}),\\
R_{3}(X)&=R_{3}(q^{H},q^{H},q^{H})z^{3}+3R_{3}(q^{H},q^{H},\bar q^{H})z^{2}\bar z\notag\\
&\quad+3R_{3}(q^{H},\bar q^{H},\bar q^{H})z\bar z^{2}
+R_{3}(\bar q^{H},\bar q^{H},\bar q^{H})\bar z^{3}+O(|z|^{4}),
\end{align}
so that
\begin{equation}
\ip{q^{H*}}{R_{2}(X)}+\ip{q^{H*}}{R_{3}(X)}
=g_{20}z^{2}+g_{11}z\bar z+g_{02}\bar z^{2}+g_{30}z^{3}+g_{21}z^{2}\bar z
+g_{12}z\bar z^{2}+g_{03}\bar z^{3}+O(|z|^{4}),
\end{equation}
with
\begin{align}
&g_{20}=\ip{q^{H*}}{R_{2}(q^{H},q^{H})},\quad
g_{11}=\ip{q^{H*}}{R_{2}(q^{H},\bar q^{H})},\quad
g_{02}=\ip{q^{H*}}{R_{2}(\bar q^{H},\bar q^{H})},\notag\\
&g_{30}=\ip{q^{H*}}{R_{3}(q^{H},q^{H},q^{H})},\notag\\
&g_{21}=\ip{q^{H*}}{R_{3}(q^{H},q^{H},\bar q^{H})}
+\ip{q^{H*}}{2R_{2}(q^{H},w_{11})+R_{2}(\bar q^{H},w_{20})} .
\end{align}
Comparing coefficients of $z^{2}$ and $z\bar z$ in the $w$-equation,
\begin{equation}
(2\ii\omega_{0}I-L)w_{20}=H_{20},\qquad -Lw_{11}=H_{11},
\end{equation}
where $H_{20}=R_{2}(q^{H},q^{H})-g_{20}q^{H}-\bar g_{02}\bar q^{H}$ and
$H_{11}=R_{2}(q^{H},\bar q^{H})-g_{11}q^{H}-\bar g_{11}\bar q^{H}$, whence
$w_{20}=(2\ii\omega_{0}I-L)^{-1}H_{20}$ and $w_{11}=-L^{-1}H_{11}$.

\begin{remark}\label{rem:w0}
Because the Hopf mode is spatially homogeneous, $H_{20}$ and $H_{11}$ are
constant in $u$ and lie in the two-dimensional space spanned by $q^{H}$ and
$\bar q^{H}$, which is exactly the center subspace of the $k=0$ block. Both
projections therefore annihilate them: $H_{20}=H_{11}=0$ and consequently
$w_{20}=w_{11}=0$. This is why $g_{21}$ reduces to
$\ip{q^{H*}}{R_{3}(q^{H},q^{H},\bar q^{H})}$.
\end{remark}

The reduced equation on the center manifold is thus
\begin{equation}\label{eq:reduced}
\dot z=\ii\omega_{0}z+\tfrac12 g_{20}z^{2}+g_{11}z\bar z
+\tfrac12 g_{02}\bar z^{2}+\tfrac12 g_{21}z^{2}\bar z+O(|z|^{4}),
\end{equation}
and, following Hassard, Kazarinoff and Wan~\cite{hassard1981theory} (the same
quantity is obtained from the invariant expression of
Kuznetsov~\cite{kuznetsov2004elements}),
\begin{equation}\label{eq:c1}
c_{1}(c_{0})=\frac{\ii}{2\omega_{0}}
\Bigl(g_{20}g_{11}-2|g_{11}|^{2}-\tfrac13|g_{02}|^{2}\Bigr)+\tfrac12 g_{21},
\qquad
\ell_{1}=\mathrm{Re}\,c_{1}(c_{0})
=\tfrac12\mathrm{Re}\,g_{21}-\frac{1}{2\omega_{0}}\mathrm{Im}\,(g_{20}g_{11}).
\end{equation}
The Hopf bifurcation is supercritical if $\ell_{1}<0$ and subcritical if
$\ell_{1}>0$.

\subsection{Numerical evaluation}\label{sec:hopfpdenum}

For $\gamma=11$, $\alpha=1$, $\delta=1$, $\beta=3.9$, $\xi=0.3$ the
biologically relevant Hopf point is
\begin{equation}
c_{0}=0.1004263750,\qquad \Es=(0.0640448876,\,1.3561030602),
\end{equation}
at which
\begin{equation}
J^{*}=\begin{pmatrix}
0.0408565543 & -0.0469521848\\
2.8424967620 & -0.0408565543
\end{pmatrix},
\end{equation}
with eigenvalues $\pm\ii\omega_{0}$,
$\omega_{0}=\sqrt{\det J^{*}}=0.3630319203$. The critical eigenvector is
\begin{equation}
q^{H}=\begin{pmatrix}1\\ 0.8701736568-7.7319494721\,\ii\end{pmatrix},
\qquad \ip{q^{H*}}{q^{H}}=1,
\end{equation}
and the normal-form coefficients are
\begin{equation}\label{eq:gvals}
\begin{aligned}
g_{20}&=\phantom{-}2.6473595630+5.5991152020\,\ii,
g_{11}&=-0.0042683804-0.2690342332\,\ii,\\
g_{02}&=-2.6558963230-5.2053936810\,\ii,
g_{21}&=-2.1729366340-3.8650870440\,\ii .
\end{aligned}
\end{equation}
Consequently
\begin{equation}
c_{1}(c_{0})=-0.0726049860-15.7509488338\,\ii,
\end{equation}
so that
\begin{equation}\label{eq:l1}
\ell_{1}=\mathrm{Re}\,c_{1}(c_{0})=-0.0726049860<0 .
\end{equation}

The Hopf bifurcation is therefore supercritical and the bifurcating periodic
solutions are stable: at $c=c_{0}$ the coexistence equilibrium
$\Es=(0.0640449,1.3561031)$ undergoes a supercritical Hopf bifurcation with
frequency $\omega_{0}$. The stability of $\Es$ is governed by the critical value
$c_{0}=0.1004263750$; $\Es$ is unstable for $0<c<c_{0}$ and locally
asymptotically stable for $c>c_{0}$. When $c$ crosses $c_{0}$, \eqref{eq:pde}
undergoes a spatially homogeneous Hopf bifurcation at $\Es$ and the bifurcating
periodic solutions are stable. Numerical illustrations on $\Om=(0,\pi)$ are shown
in Figs.~\ref{fig:4}--\ref{fig:6}.

\begin{figure}[ht!]
\centering
\begin{minipage}{0.48\linewidth}\figph{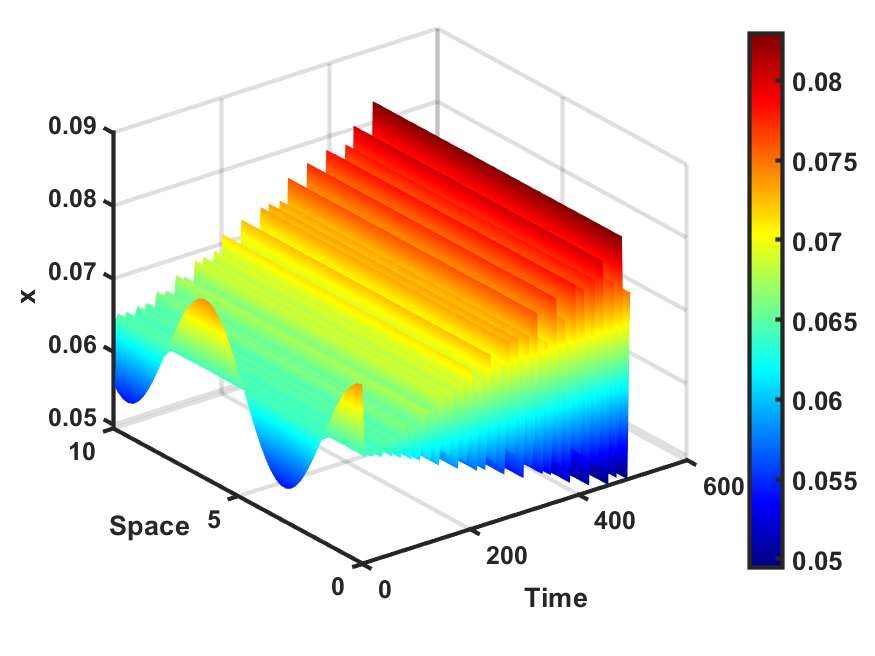}{45mm}\end{minipage}\hfill
\begin{minipage}{0.48\linewidth}\figph{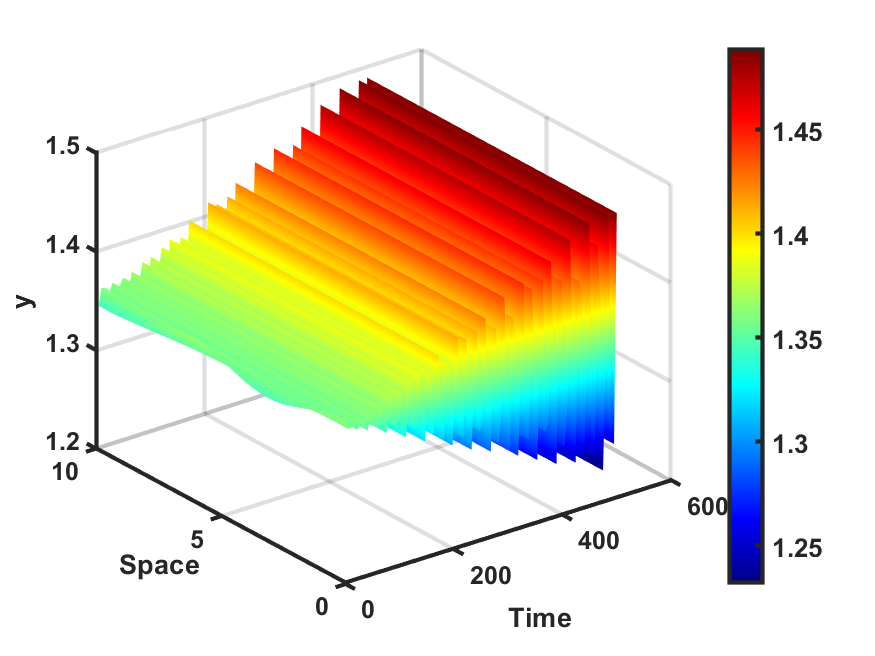}{45mm}\end{minipage}
\caption{\label{fig:4}%
$\gamma=11$, $\alpha=1$, $\delta=1$, $\beta=3.9$, $\xi=0.3$, $d_{1}=d_{2}=1$,
$c=0.1004263750$. Initial data
$(0.0640448876+0.01\cos u,\ 1.3561030602+0.01\cos u)$. The constant state
$\Es=(0.0640448876,1.3561030602)$ is unstable and \eqref{eq:pde} approaches a
stable spatially homogeneous periodic solution surrounding $\Es$.}
\alttext{Two three-dimensional surface plots of prey and of predator density over space and time. Both fields oscillate with a fixed period and remain flat across space at every instant, the amplitude approaching a constant value.}
\end{figure}

\begin{figure}[ht!]
\centering
\begin{minipage}{0.48\linewidth}\figph{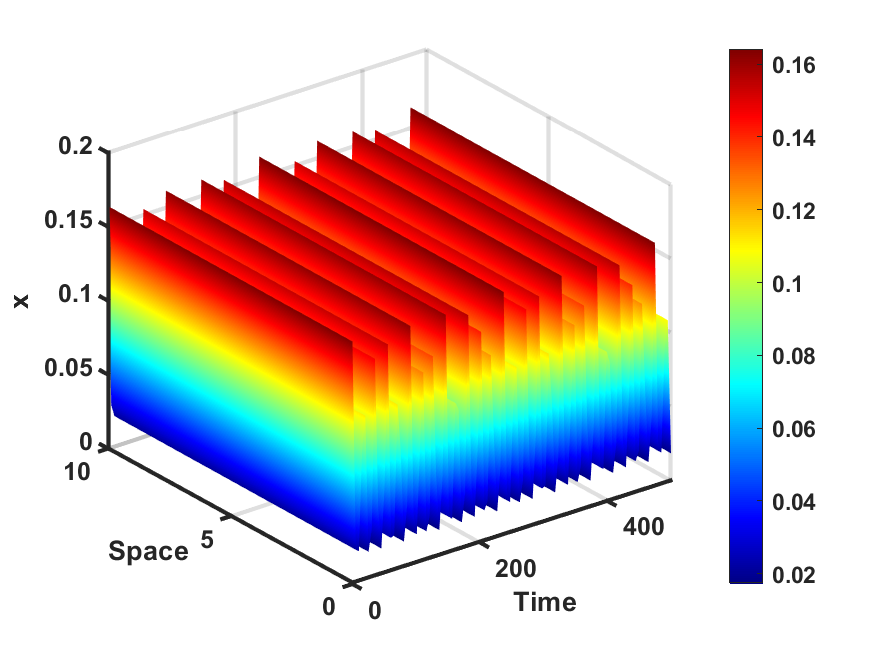}{45mm}\end{minipage}\hfill
\begin{minipage}{0.48\linewidth}\figph{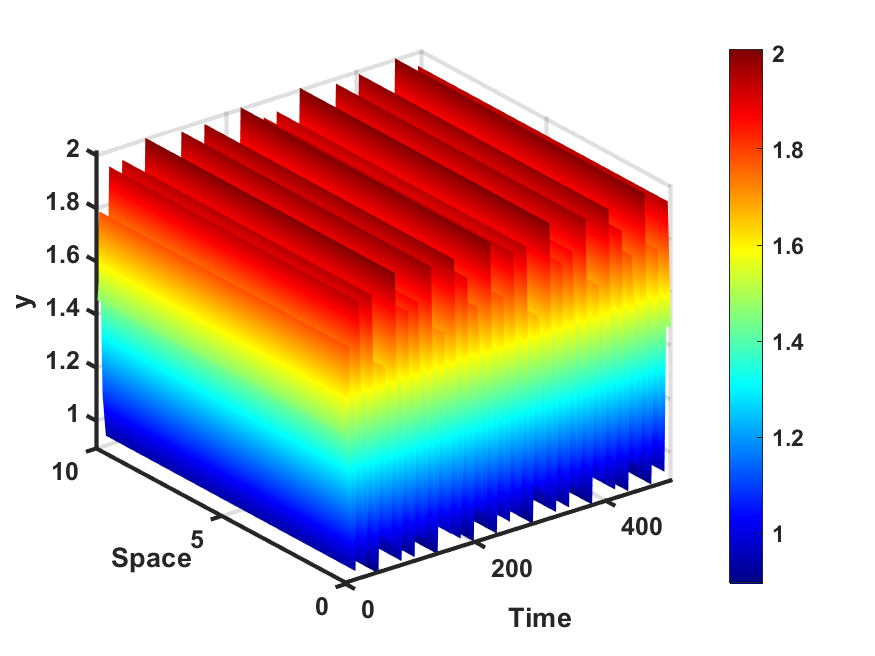}{45mm}\end{minipage}
\caption{\label{fig:5}%
As in Fig.~\ref{fig:4} but with spatially uniform initial data
$(0.0640448876+0.1,\ 1.3561030602+0.1)$.}
\alttext{Two three-dimensional surface plots of prey and of predator density over space and time, started from data with no spatial variation. Both settle onto the same spatially flat oscillation reached in the previous figure.}
\end{figure}

\begin{figure}[ht!]
\centering
\begin{minipage}{0.48\linewidth}\figph{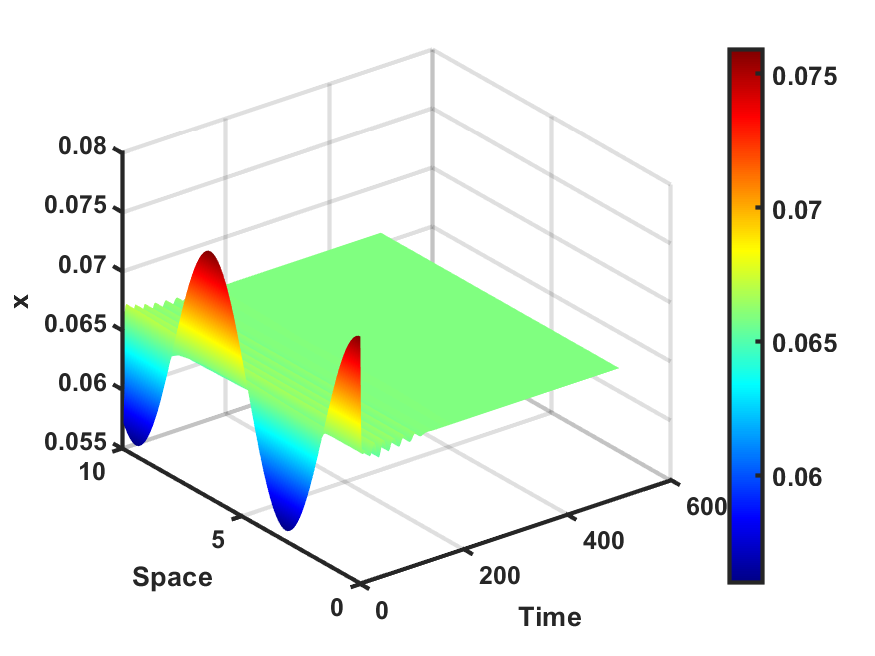}{45mm}\end{minipage}\hfill
\begin{minipage}{0.48\linewidth}\figph{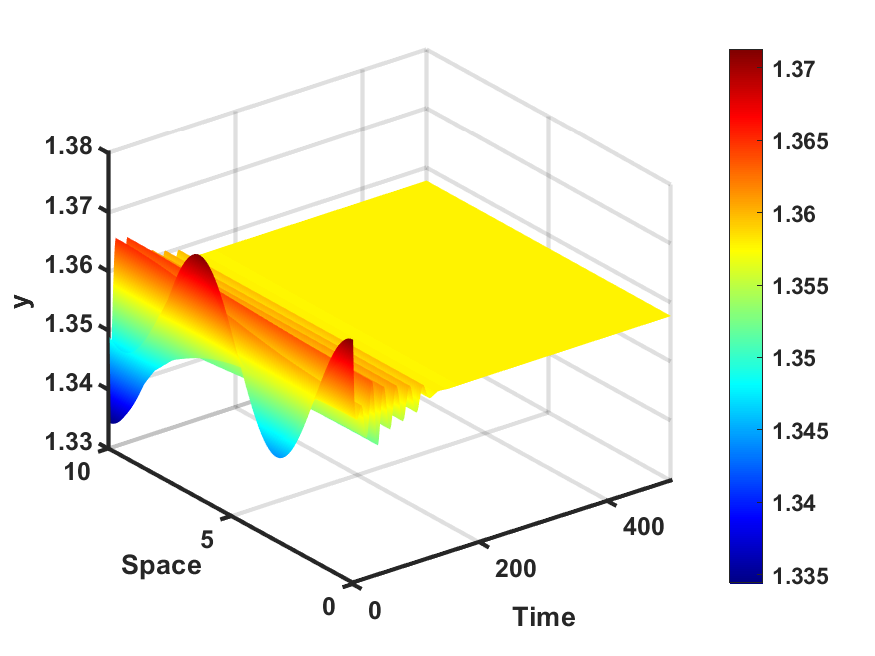}{45mm}\end{minipage}
\caption{\label{fig:6}%
$\gamma=11$, $\alpha=1$, $\delta=1$, $\beta=3.9$, $\xi=0.3$, $d_{1}=d_{2}=1$,
$c=0.11>c_{0}$. Initial data
$(0.0659+0.01\cos u,\ 1.3577+0.01\cos u)$. The constant state
$\Es=(0.0659315,1.3577444)$ is locally asymptotically stable.}
\alttext{Two three-dimensional surface plots of prey and of predator density over space and time. Initial oscillations decay steadily and both fields flatten onto a constant value, indicating return to the uniform steady state.}
\end{figure}

%% =====================================================================
\section{Turing--Hopf bifurcation}\label{sec:th}
%% =====================================================================

By Sec.~\ref{sec:hopf-pde}, \eqref{eq:pde} undergoes a spatially homogeneous
Hopf bifurcation when
\begin{equation}
c=c^{*}=-\frac{a_{1}}{b_{2}}=\frac{F'(\xs)p(\xs)}{\xi \ys},
\end{equation}
so that the Hopf curve in the $(d_{1},c)$-plane is the horizontal line
\begin{equation}
\mathcal{H}_{c}:\ c=c^{*}=\frac{F'(\xs)p(\xs)}{\xi \ys}.
\end{equation}
Taking $d_{1}$ as the bifurcation parameter for the Turing instability and
solving $D_{k}=0$ in \eqref{eq:Dk} gives the family of Turing curves
\begin{equation}\label{eq:turingcurve}
\mathcal{T}_{d_{1}}:\quad
d_{1}=d_{1}^{k}=\frac{a_{1}d_{2}z_{k}-(ca_{1}b_{2}-a_{2}b_{1})}
{z_{k}(d_{2}z_{k}-cb_{2})} .
\end{equation}
Since only positive diffusivities are admissible, define
\begin{equation}\label{eq:Qdef}
Q=\{k\in\N:\ d_{1}^{k}>0\},
\end{equation}
and, when $Q\ne\emptyset$, let $k^{*}=\arg\max_{k\in Q}d_{1}^{k}$ be the first
mode to become unstable as $d_{1}$ decreases.

Writing $\lambda_{1}(c)=\alpha_{1}(c)+\ii\beta_{1}(c)$ with
$\alpha_{1}(c^{*})=0$ and $\beta_{1}(c^{*})=\omega_{0}$, one has
\begin{equation}\label{eq:transhopf}
b_{2}<0,\qquad
\frac{\dd\,\mathrm{Re}\,\lambda_{1}(c)}{\dd c}\Big|_{c=c^{*}}=\frac{b_{2}}{2}<0,
\end{equation}
so the Hopf bifurcation is transversal. For the Turing bifurcation,
\begin{equation}\label{eq:transturing}
\frac{\partial D_{k}}{\partial d_{1}}\bigg|_{(d_{1},c)=(d_{1}^{k^{*}},c^{*})}
=z_{k^{*}}\bigl(d_{2}z_{k^{*}}-c^{*}b_{2}\bigr)>0,
\end{equation}
so the Turing bifurcation is transversal as well. We therefore have:

\begin{theorem}\label{thm:th}
Let $(R_{2})$ hold. If $Q=\emptyset$, then \eqref{eq:pde} does not exhibit a
Turing--Hopf bifurcation. If $Q\ne\emptyset$, then \eqref{eq:pde} undergoes a
Turing--Hopf bifurcation at $\Es$ when $(d_{1},c)=(d_{1}^{k^{*}},c^{*})$.
\end{theorem}

\begin{proof}
If $Q=\emptyset$, no admissible wavenumber satisfies the Turing instability
conditions with a positive diffusivity, so no Turing curve
$\mathcal{T}_{d_{1}}$ exists in the $(d_{1},c)$-plane; the Hopf curve
$\mathcal{H}_{c}$ has no intersection with a Turing curve and no Turing--Hopf
bifurcation can occur. If $Q\ne\emptyset$ there is $k^{*}\in\N$ with
$d_{1}^{k^{*}}>0$, and $\mathcal{T}_{d_{1}}$ meets $\mathcal{H}_{c}$ at
$(d_{1}^{k^{*}},c^{*})$; this is the generic Turing--Hopf configuration
described in Refs.~\cite{song2016turing,jiang2020formulation}. Moreover
$T_{k}<0$ for $k\ge1$ and $D_{k}>0$ for all $k\in\N\setminus\{k^{*}\}$, so that
apart from the purely imaginary pair belonging to the Hopf mode $k=0$ and the
simple zero eigenvalue belonging to the critical Turing mode $k=k^{*}$, every
eigenvalue of \eqref{eq:char} has negative real part. Transversality follows
from \eqref{eq:transhopf} and \eqref{eq:transturing}.
\end{proof}

\subsection{Normal form at the Turing--Hopf point}\label{sec:thnf}

Set $d_{1}=d_{1}^{k^{*}}+\mu_{1}$ and $c=c^{*}+\mu_{2}$, so that \eqref{eq:pde}
undergoes the Turing--Hopf bifurcation at $\mu_{1}=\mu_{2}=0$. Translating $\Es$
to the origin and writing $X(t)=(x(t),y(t))$, the system
becomes~\cite{jiang2020formulation}
\begin{equation}\label{eq:Xmu}
\frac{\dd X(t)}{\dd t}=D(\mu)\Delta X(t)+L(\mu)X(t)+G\bigl(X(t),\mu\bigr),
\end{equation}
where
\begin{equation}
D(\mu)=\begin{pmatrix}d_{1}^{k^{*}}+\mu_{1}&0\\0&d_{2}\end{pmatrix}
=D_{0}+D_{1}(\mu_{1}),\qquad
L(\mu)=\begin{pmatrix}a_{1}&a_{2}\\ b_{1}&(c^{*}+\mu_{2})b_{2}\end{pmatrix}
=L_{0}+L_{1}(\mu_{2}),
\end{equation}
and $G$ collects the nonlinear terms, with second- and third-order Fr\'echet
derivatives
\begin{equation}
P(\phi,\phi)=\begin{pmatrix}a_{20}\phi_{1}^{2}+a_{11}\phi_{1}\phi_{2}\\
b_{20}\phi_{1}^{2}+b_{02}\phi_{2}^{2}+b_{11}\phi_{1}\phi_{2}\end{pmatrix},
\qquad
Q(\phi,\phi,\phi)=\begin{pmatrix}a_{30}\phi_{1}^{3}+a_{21}\phi_{1}^{2}\phi_{2}\\
b_{30}\phi_{1}^{3}+b_{21}\phi_{1}^{2}\phi_{2}\end{pmatrix},
\end{equation}
$\phi=(\phi_{1},\phi_{2})^{T}$. The characteristic matrix of the linearised
system is
\begin{equation}
\Delta_{k}(\lambda)=\begin{pmatrix}
\lambda+d_{1}z_{k}-a_{1}&-a_{2}\\ -b_{1}&\lambda+d_{2}z_{k}-c^{*}b_{2}
\end{pmatrix}.
\end{equation}
At the Turing--Hopf point $\lambda=0$ is a simple eigenvalue of
$\Delta_{k^{*}}(0)$ and $\lambda=\pm\ii\omega_{0}$ are simple eigenvalues of
$\Delta_{0}(\lambda)$, all remaining eigenvalues having negative real parts. The
right eigenvectors are
\begin{equation}
\phi_{1}=\begin{pmatrix}1\\[2pt]\dfrac{d_{1}^{k^{*}}z_{k^{*}}-a_{1}}{a_{2}}\end{pmatrix},
\qquad
\phi_{2}=\begin{pmatrix}1\\[2pt]\dfrac{\ii\omega_{0}-a_{1}}{a_{2}}\end{pmatrix},
\end{equation}
and the adjoint eigenvectors, normalized by $\ip{\psi_{1}}{\phi_{1}}=1$,
$\ip{\psi_{2}}{\phi_{2}}=1$, $\ip{\psi_{2}}{\bar\phi_{2}}=0$, are
\begin{equation}
\psi_{1}=\frac{1}{1+\dfrac{(a_{1}-d_{1}^{k^{*}}z_{k^{*}})^{2}}{a_{2}b_{1}}}
\begin{pmatrix}1\\[2pt]\dfrac{d_{1}^{k^{*}}z_{k^{*}}-a_{1}}{b_{1}}\end{pmatrix},
\qquad
\psi_{2}=\frac{1}{1+\dfrac{(a_{1}-\ii\omega_{0})^{2}}{a_{2}b_{1}}}
\begin{pmatrix}1\\[2pt]-\dfrac{a_{1}+\ii\omega_{0}}{b_{1}}\end{pmatrix}.
\end{equation}
With $\Phi=(\phi_{1},\phi_{2},\bar\phi_{2})$ and
$\Psi=(\psi_{1},\psi_{2},\bar\psi_{2})^{T}$ the bi-orthogonality relations
$\ip{\psi_{i}}{\phi_{j}}=\delta_{ij}$ hold.

Following the Turing--Hopf normal form theory of Jiang, An and
Shi~\cite{jiang2020formulation}, the dynamics on the three-dimensional center
manifold are governed by
\begin{equation}\label{eq:thnf}
\left\{
\begin{aligned}
\dot z_{1}&=a_{1}(\mu)z_{1}+a_{200}z_{1}^{2}+a_{011}z_{2}\bar z_{2}
+a_{300}z_{1}^{3}+a_{111}z_{1}z_{2}\bar z_{2}+\text{h.o.t.},\\
\dot z_{2}&=\ii\omega_{0}z_{2}+b_{2}(\mu)z_{2}+b_{110}z_{1}z_{2}
+b_{210}z_{1}^{2}z_{2}+b_{021}z_{2}^{2}\bar z_{2}+\text{h.o.t.},\\
\dot{\bar z}_{2}&=-\ii\omega_{0}\bar z_{2}+\overline{b_{2}(\mu)}\bar z_{2}
+\overline{b_{110}}z_{1}\bar z_{2}+\overline{b_{210}}z_{1}^{2}\bar z_{2}
+\overline{b_{021}}z_{2}\bar z_{2}^{2}+\text{h.o.t.}
\end{aligned}\right.
\end{equation}
For $k_{H}=0$, $k^{*}\ne0$ and Neumann conditions on $\Om=(0,l\pi)$,
Ref.~\cite{jiang2020formulation} gives
\begin{align}
a_{1}(\mu)&=\tfrac12\psi_{1}\bigl(L_{1}(\mu)\phi_{1}-z_{k^{*}}D_{1}(\mu)\phi_{1}\bigr),
\qquad
b_{2}(\mu)=\tfrac12\psi_{2}L_{1}(\mu)\phi_{2},\notag\\
a_{200}&=a_{011}=b_{110}=0,\notag\\
a_{300}&=\tfrac14\ip{\psi_{1}}{Q(\phi_{1},\phi_{1},\phi_{1})}
+\frac{1}{\omega_{0}}\mathrm{Re}\bigl(\ii\ip{\psi_{2}}{P(\phi_{1},\phi_{2})}\bigr)
\ip{\psi_{1}}{P(\phi_{1},\phi_{1})}\notag\\
&\quad+\Bigl\langle\psi_{1},P\bigl(\phi_{1},h_{200}^{0}+\tfrac{1}{\sqrt2}h_{200}^{2k^{*}}\bigr)\Bigr\rangle,\notag\\
a_{111}&=\ip{\psi_{1}}{Q(\phi_{1},\phi_{2},\bar\phi_{2})}
+\frac{2}{\omega_{0}}\mathrm{Re}\bigl(\ii\ip{\psi_{2}}{P(\phi_{1},\phi_{2})}\bigr)
\ip{\psi_{1}}{P(\phi_{2},\bar\phi_{2})}\notag\\
&\quad+\Bigl\langle\psi_{1},P(\phi_{1},h_{011}^{0})
+\tfrac{1}{\sqrt2}P(\phi_{1},h_{011}^{2k^{*}})
+P(\phi_{2},h_{101}^{k^{*}})+P(\bar\phi_{2},h_{110}^{k^{*}})\Bigr\rangle,\notag\\
b_{210}&=\tfrac12\psi_{2}Q(\phi_{1},\phi_{1},\phi_{2})
+\psi_{2}\bigl[P(\phi_{1},h_{110}^{k^{*}})+P(\phi_{2},h_{200}^{0})\bigr]\notag\\
&\quad+\frac{1}{2\ii\omega_{0}}\psi_{2}
\Bigl[2\bigl(\psi_{1}P(\phi_{1},\phi_{1})\bigr)P(\phi_{1},\phi_{2})
+\bigl(-P(\phi_{2},\phi_{2})\psi_{2}+P(\phi_{2},\bar\phi_{2})\bar\psi_{2}\bigr)
P(\phi_{1},\phi_{1})\Bigr],\notag\\
b_{021}&=\tfrac12\psi_{2}Q(\phi_{2},\phi_{2},\bar\phi_{2})
+\psi_{2}\bigl[P(\phi_{2},h_{011}^{0})+P(\bar\phi_{2},h_{020}^{0})\bigr]\notag\\
&\quad+\frac{1}{4\ii\omega_{0}}\psi_{2}
\Bigl[\tfrac23\bigl(\bar\psi_{2}P(\bar\phi_{2},\bar\phi_{2})\bigr)P(\phi_{2},\phi_{2})
+\bigl(-2\psi_{2}P(\phi_{2},\phi_{2})+4\bar\psi_{2}P(\phi_{2},\bar\phi_{2})\bigr)
P(\phi_{2},\bar\phi_{2})\Bigr],
\end{align}
with auxiliary functions
\begin{align}
h_{200}^{0}&=-\tfrac12 L_{0}^{-1}P(\phi_{1},\phi_{1})
+\frac{1}{2\ii\omega_{0}}(\phi_{2}\psi_{2}-\bar\phi_{2}\bar\psi_{2})P(\phi_{1},\phi_{1}),\notag\\
h_{200}^{2k^{*}}&=-\frac{1}{2\sqrt2}\bigl(L_{0}-4D_{0}z_{k^{*}}\bigr)^{-1}P(\phi_{1},\phi_{1}),\notag\\
h_{011}^{0}&=-L_{0}^{-1}P(\phi_{2},\bar\phi_{2})
+\frac{1}{\ii\omega_{0}}(\phi_{2}\psi_{2}-\bar\phi_{2}\bar\psi_{2})P(\phi_{2},\bar\phi_{2}),\notag\\
h_{020}^{0}&=\tfrac12(2\ii\omega_{0}I-L_{0})^{-1}P(\phi_{2},\phi_{2})
-\frac{1}{2\ii\omega_{0}}\Bigl(\phi_{2}\psi_{2}+\tfrac13\bar\phi_{2}\bar\psi_{2}\Bigr)P(\phi_{2},\phi_{2}),\notag\\
h_{110}^{k^{*}}&=\bigl(\ii\omega_{0}I-L_{0}+D_{0}z_{k^{*}}\bigr)^{-1}P(\phi_{1},\phi_{2})
-\frac{1}{\ii\omega_{0}}\phi_{1}\psi_{1}P(\phi_{1},\phi_{2}),\notag\\
h_{101}^{k^{*}}&=\overline{h_{110}^{k^{*}}},\qquad
h_{002}^{0}=\overline{h_{020}^{0}},\qquad h_{011}^{2k^{*}}=0 .
\end{align}
The qualitative dynamics near the Turing--Hopf point are governed by
\eqref{eq:thnf}, and in particular:

\begin{theorem}[Ref.~\cite{li2025turing}]\label{thm:hopfpitch}
If $a_{300}\ne0$, $a_{111}\ne0$, $\mathrm{Re}(b_{210})\ne0$,
$\mathrm{Re}(b_{021})\ne0$ and
$a_{300}\mathrm{Re}(b_{021})-a_{111}\mathrm{Re}(b_{210})\ne0$, then, provided
$(R_{2})$ holds and $Q\ne\emptyset$, a Hopf--pitchfork bifurcation occurs at
$\Es$ for $(c,d_{1})=(c^{*},d_{1}^{k^{*}})$.
\end{theorem}

\subsection{Numerical study of the Turing--Hopf bifurcation}\label{sec:thnum}

We fix
\begin{equation}
\gamma=11,\quad \alpha=1,\quad \delta=1,\quad \beta=3.9,\quad \xi=0.3,\quad d_{2}=1,
\end{equation}
and take $\Om=(0,l\pi)$ with $l=2$, so that $z_{k}=k^{2}/l^{2}=k^{2}/4$. The
parameters $c$ and $d_{1}$ are the bifurcation parameters.

The coexistence equilibrium and its Jacobian are
\begin{equation}
\Es=(0.0640448876,\,1.3561030602),\qquad
J(\Es)=\begin{pmatrix}
0.0408565543& -0.0469521848\\
2.8424967620&-0.0408565543
\end{pmatrix},
\end{equation}
the trace vanishing identically at $c=c^{*}$, as it must. The Hopf condition
$a_{1}+cb_{2}=0$ gives
\begin{equation}
c^{*}=0.1004263750,
\end{equation}
so that the homogeneous mode $k=0$ carries a pair of purely imaginary
eigenvalues $\lambda_{2,3}=\pm\ii\omega_{0}$, $\omega_{0}=0.3630319203$.

Evaluating \eqref{eq:turingcurve} at $c=c^{*}$ mode by mode gives
\begin{equation}\label{eq:d1k}
\begin{aligned}
d_{1}^{4}&=0.0019571372, & d_{1}^{5}&=0.0031426255, & d_{1}^{6}&=0.0028993911,\\
d_{1}^{7}&=0.0024488130, & d_{1}^{8}&=0.0020335288, &
d_{1}^{k}&<0\ \ (k\le3),
\end{aligned}
\end{equation}
so that $Q=\{k\ge4\}$ and the maximum is attained at
\begin{equation}
k^{*}=5,\qquad z_{k^{*}}=\Bigl(\frac{5}{2}\Bigr)^{2}=6.25,
\end{equation}
with
\begin{equation}\label{eq:d1TH}
d_{1}^{\mathrm{TH}}=0.0031426255 .
\end{equation}
The Turing--Hopf point in the $(c,d_{1})$-plane is therefore
\begin{equation}
(c^{*},d_{1}^{\mathrm{TH}})=(0.1004263750,\,0.0031426255),
\end{equation}
at which the homogeneous mode $k=0$ has a pair of purely imaginary eigenvalues
while the mode $k=5$ has a simple zero eigenvalue: the Hopf and Turing
instabilities occur simultaneously, confirming the existence of a Turing--Hopf
bifurcation at $\Es$. Transversality is immediate from \eqref{eq:transhopf},
$\dd\,\mathrm{Re}\,\lambda_{1}/\dd c=b_{2}/2=-0.2034155$, and
\eqref{eq:transturing},
$\partial D_{5}/\partial d_{1}=z_{5}(z_{5}-c^{*}b_{2})=39.3178537$.

At the bifurcation point the critical eigenvectors are
\begin{equation}
\phi_{1}=\begin{pmatrix}1\\ 0.4518457\end{pmatrix},\qquad
\phi_{2}=\begin{pmatrix}1\\ 0.8701737-7.7319495\,\ii\end{pmatrix}.
\end{equation}
Two of the normal-form coefficients can be pinned down independently of the full
center-manifold computation, and both are informative.

First, setting $z_{2}=0$ in \eqref{eq:thnf} leaves the pure Turing branch
$\dot z_{1}=a_{1}(\mu)z_{1}+a_{300}z_{1}^{3}$, whose cubic coefficient is the
Turing normal-form coefficient of Sec.~\ref{sec:turingnf} evaluated at $k^{*}=5$
and $d_{1}=d_{1}^{\mathrm{TH}}$. Direct computation gives
$g_{3}^{(5)}=-0.9183721<0$ in the normalization $\phi_{1,1}=1$,
$\ip{\psi_{1}}{\phi_{1}}=1$; the sign of this quantity is invariant under
rescaling of the critical eigenvector, so
\begin{equation}
a_{300}<0,
\end{equation}
and the pure Turing branch emerging from the Turing--Hopf point is
supercritical.

Second, setting $z_{1}=0$ leaves the pure Hopf branch
$\dot z_{2}=\ii\omega_{0}z_{2}+b_{021}z_{2}^{2}\bar z_{2}$, which is the
spatially homogeneous Hopf normal form of Sec.~\ref{sec:hopf-pde}. Hence
$b_{021}=c_{1}(c_{0})$ in the normalization used there, so that
$\mathrm{Re}(b_{021})<0$.

Carrying out the full three-dimensional center-manifold reduction at
$(c^{*},d_{1}^{\mathrm{TH}})$, with the critical eigenfunctions
$\phi_{1}\cos(k^{*}u/l)$, $\phi_{2}$ and $\bar\phi_{2}$, the quadratic
center-manifold coefficients computed mode by mode, and the quadratic terms of
the reduced field removed by a near-identity transformation, gives
\begin{equation}\label{eq:nfcoef}
\begin{aligned}
a_{300}&=-0.9183721221, &\qquad
a_{111}&=\phantom{-}0.1097666894,\\
b_{210}&=\phantom{-}0.9608730216-0.7470167113\,\ii, &\qquad
b_{021}&=-0.0726049860-15.7509488338\,\ii,
\end{aligned}
\end{equation}
and therefore
\begin{equation}
a_{300}\,\mathrm{Re}(b_{021})-a_{111}\,\mathrm{Re}(b_{210})=-0.0387934555\ne0 .
\end{equation}
Two internal consistency checks support \eqref{eq:nfcoef}. The coefficient
$b_{021}$ obtained from the three-dimensional reduction agrees to fourteen
digits with $c_{1}(c_{0})$ computed independently in
Sec.~\ref{sec:hopfpdenum} from Hassard's formula; and $a_{300}$ agrees to ten
digits with the pure Turing coefficient $g_{3}^{(5)}$ of
Sec.~\ref{sec:turingnf}, even though the two computations treat the spatially
uniform component of the quadratic center-manifold term differently: in the pure
Turing reduction it is retained in $w_{20}$, whereas at the Turing--Hopf point it
is projected onto the Hopf center subspace and re-enters through the normal-form
transformation.

The corresponding linear coefficients follow from the eigenvalue derivatives at
the codimension-two point,
\begin{equation}
\frac{\partial\lambda_{T}}{\partial d_{1}}=-6.2711525,\quad
\frac{\partial\lambda_{T}}{\partial c}=0.0013766,\quad
\frac{\partial\,\mathrm{Re}\,\lambda_{H}}{\partial c}=\frac{b_{2}}{2}=-0.2034155,\quad
\frac{\partial\,\mathrm{Im}\,\lambda_{H}}{\partial c}=-0.0228942,
\end{equation}
which give, in the scaling of Ref.~\cite{jiang2020formulation},
\begin{equation}
a_{1}(\mu)=-3.1355763\,\mu_{1}+0.0006883\,\mu_{2},\qquad
b_{2}(\mu)=\bigl(-0.1017078-0.0114471\,\ii\bigr)\mu_{2}.
\end{equation}

Since $a_{300}<0$, $a_{111}\ne0$, $\mathrm{Re}(b_{210})\ne0$,
$\mathrm{Re}(b_{021})<0$ and
$a_{300}\mathrm{Re}(b_{021})\ne a_{111}\mathrm{Re}(b_{210})$,
Theorem~\ref{thm:hopfpitch} applies: a Hopf--pitchfork bifurcation occurs at
$\Es$ for $(c,d_{1})=(0.1004263750,\,0.0031426255)$, and both primary branches
issuing from the codimension-two point, the spatially uniform periodic orbit and
the pair of stationary $\cos(5u/2)$ patterns, are supercritical.

Direct integration of \eqref{eq:pde} on $\Om=(0,2\pi)$ with $400$ grid points
reproduces the predicted scenario. At $c=c^{*}$ and
$d_{1}=0.0033>d_{1}^{\mathrm{TH}}$ a small perturbation of $\Es$ neither decays
nor grows appreciably, consistent with the marginal Hopf mode. At $c=c^{*}$ and
$d_{1}=0.0030<d_{1}^{\mathrm{TH}}$ the mode $\cos(5u/2)$ grows to a finite
amplitude and the solution remains time dependent, i.e.\ the system settles on a
spatially heterogeneous, temporally oscillating state, which is the signature of
the Turing--Hopf interaction and a state that neither instability produces alone.

Figures~\ref{fig:7}--\ref{fig:9} illustrate the dynamics on the Turing side of
the codimension-two point. At $c=0.15$ the coexistence equilibrium is
$\Es=(0.0739220,\,1.3646890)$ and since $c=0.15>c^{*}$ the homogeneous mode is
stable there, while the Turing threshold at that value of $c$ is
$d_{1}^{k=5}=0.0036298>0.0025$, so the solutions shown are spatially
heterogeneous and, as our integrations confirm, time dependent rather than
stationary.

\begin{figure}[ht!]
\centering
\begin{minipage}{0.48\linewidth}\figph{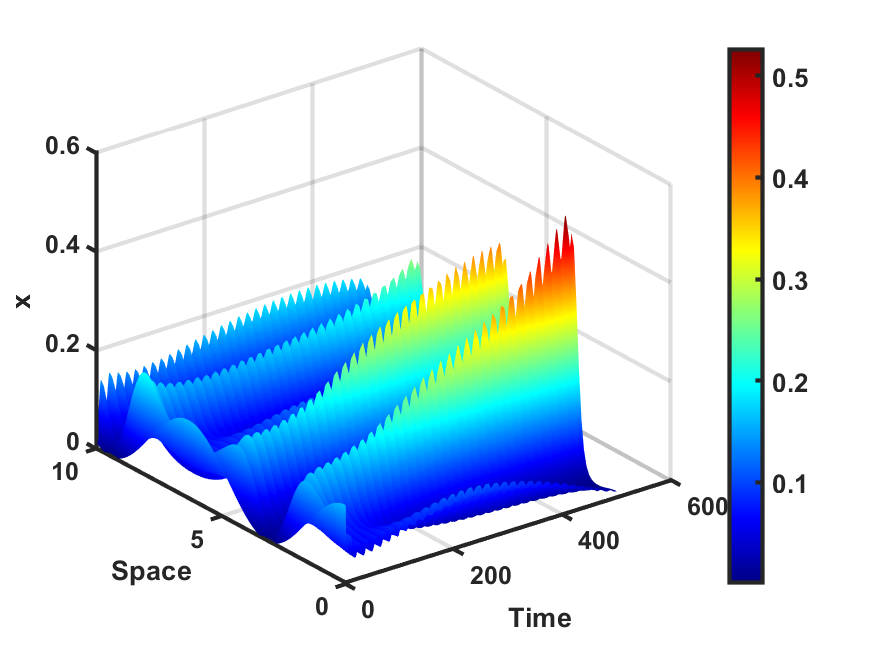}{42mm}\end{minipage}\hfill
\begin{minipage}{0.48\linewidth}\figph{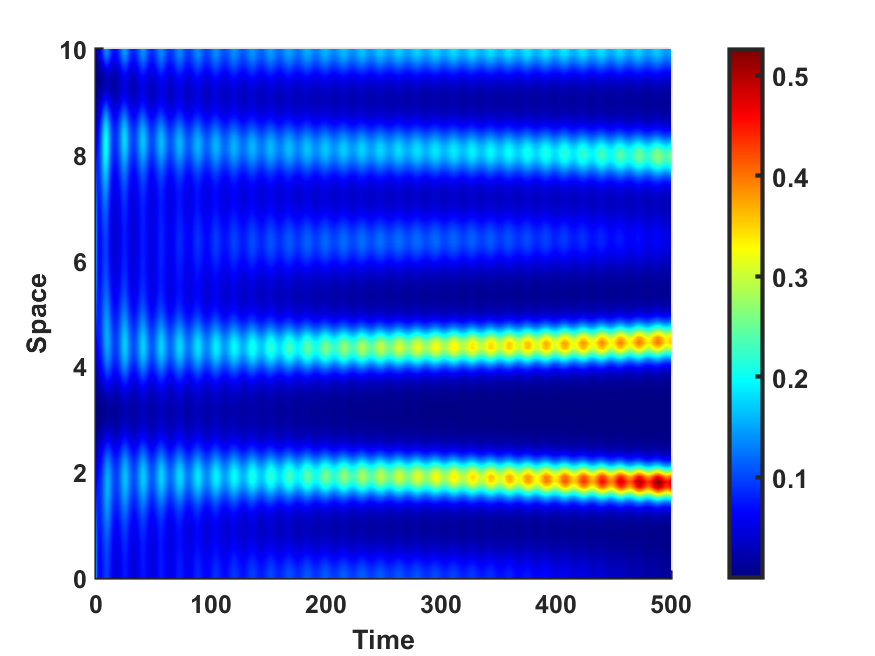}{42mm}\end{minipage}\\[6pt]
\begin{minipage}{0.48\linewidth}\figph{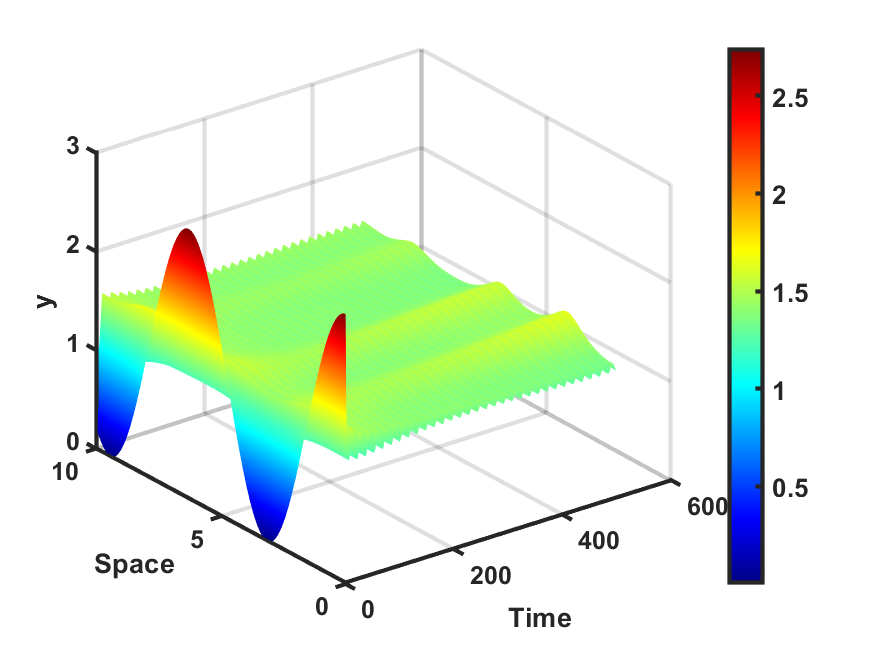}{42mm}\end{minipage}\hfill
\begin{minipage}{0.48\linewidth}\figph{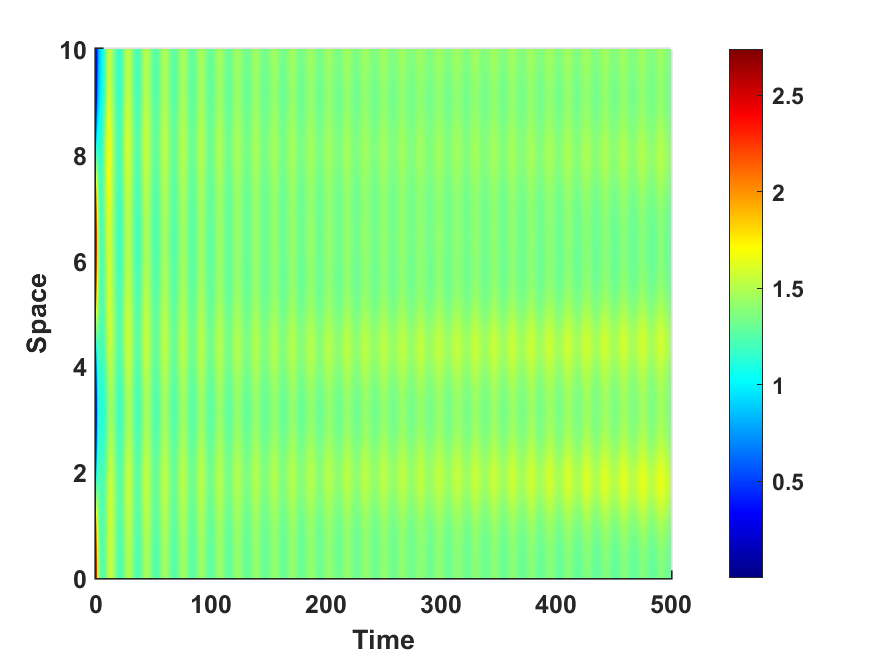}{42mm}\end{minipage}
\caption{\label{fig:7}%
$\gamma=11$, $\alpha=1$, $\delta=1$, $\beta=3.9$, $\xi=0.3$,
$d_{1}=0.0025$, $d_{2}=1$, $c=0.15$, $\Om=(0,2\pi)$. Initial data
$\bigl(x_{0}(u),y_{0}(u)\bigr)=\bigl(0.0800(1+\cos u),\,1.3700(1+\cos u)\bigr)$.
The constant state $\Es=(0.0739220,1.3646890)$ is unstable and
\eqref{eq:pde} develops a spatially heterogeneous time-periodic state.}
\alttext{Four panels: three-dimensional surfaces of prey and of predator density over space and time, each with a companion space-time heat map in which time runs horizontally and position vertically. Both fields develop a repeating spatial pattern whose amplitude keeps oscillating in time.}
\end{figure}

\begin{figure}[ht!]
\centering
\begin{minipage}{0.48\linewidth}\figph{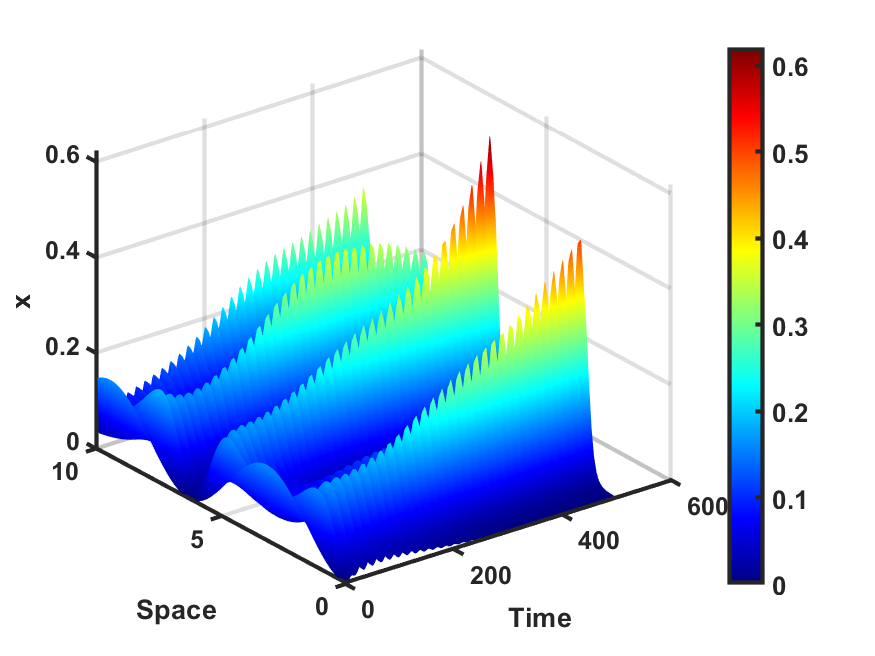}{42mm}\end{minipage}\hfill
\begin{minipage}{0.48\linewidth}\figph{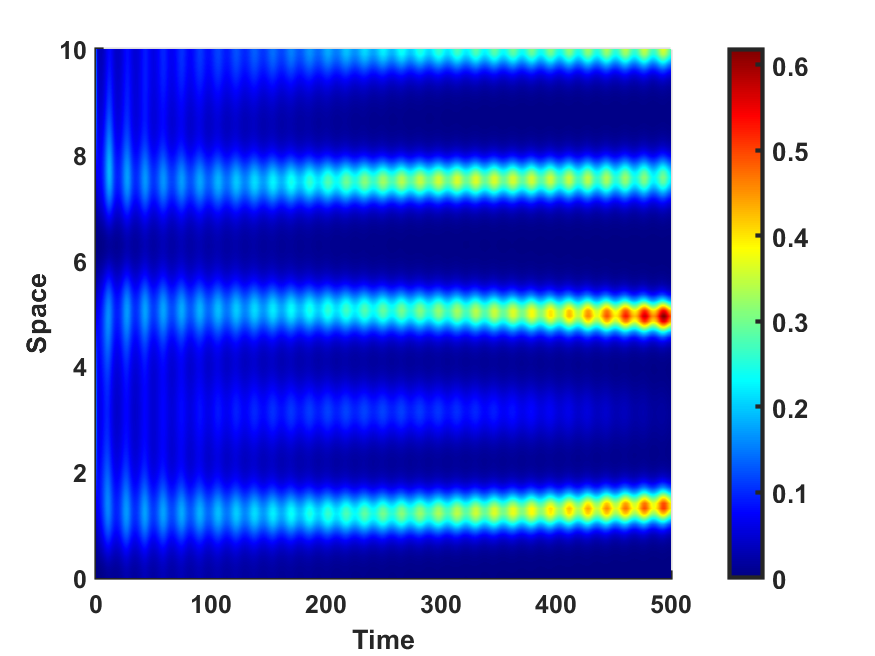}{42mm}\end{minipage}\\[6pt]
\begin{minipage}{0.48\linewidth}\figph{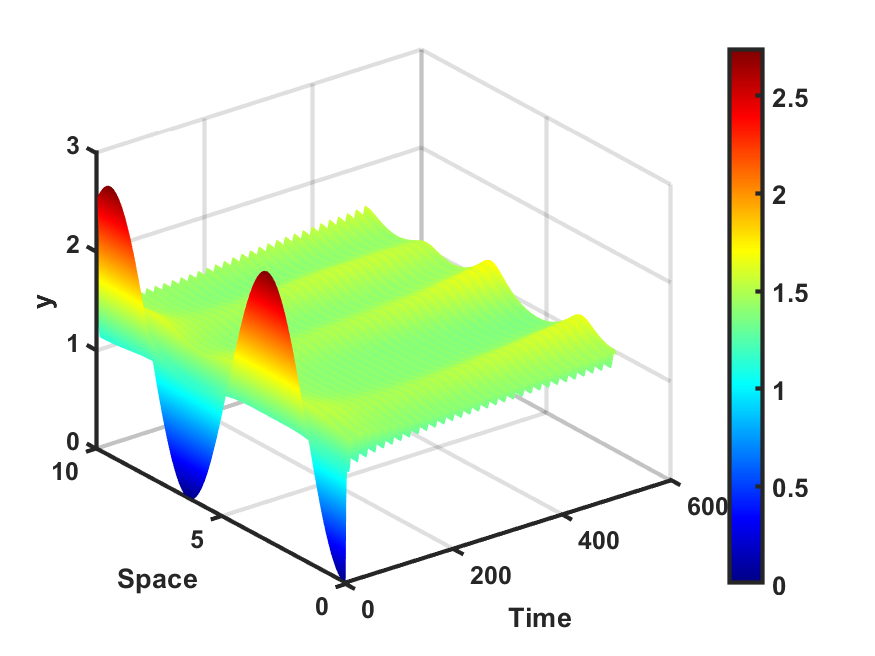}{42mm}\end{minipage}\hfill
\begin{minipage}{0.48\linewidth}\figph{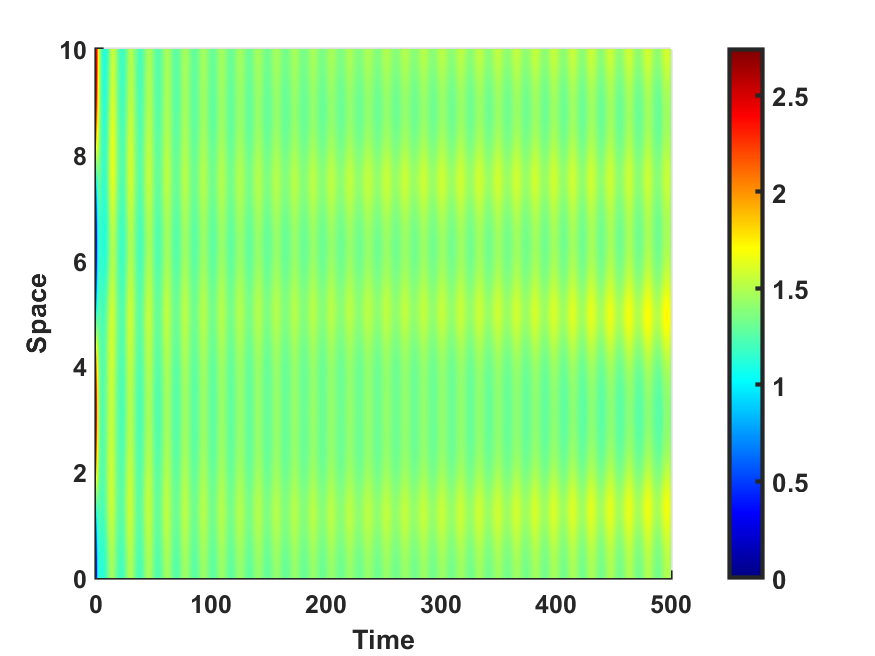}{42mm}\end{minipage}
\caption{\label{fig:8}%
As in Fig.~\ref{fig:7} with initial data
$\bigl(0.0800(1-\cos u),\,1.3700(1-\cos u)\bigr)$.}
\alttext{Four panels of the same type as the previous figure, started from the complementary spatial perturbation. The banded heat maps show that the repeating spatial pattern and its continuing temporal oscillation are reached from either perturbation.}
\end{figure}

\begin{figure}[ht!]
\centering
\begin{minipage}{0.48\linewidth}\figph{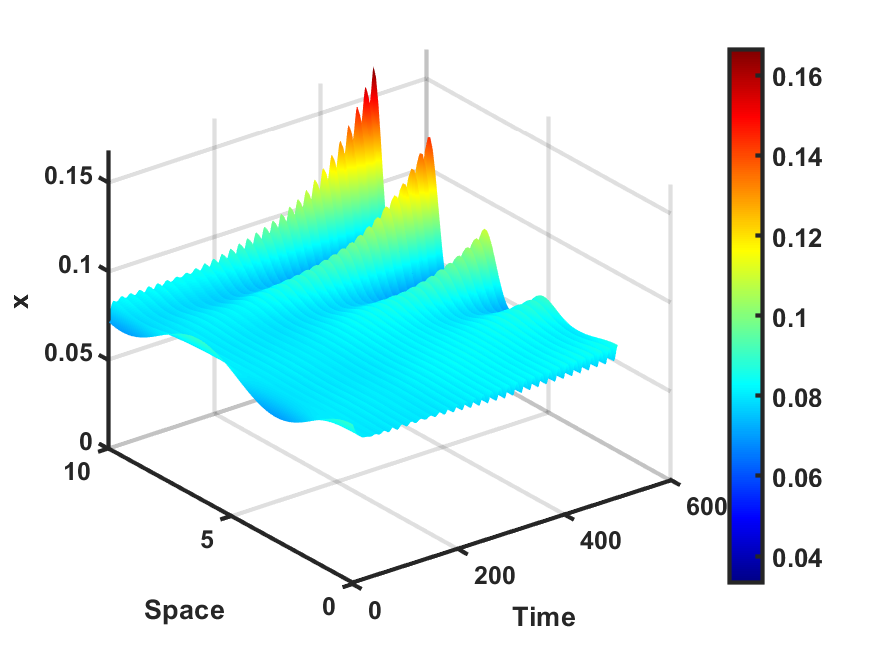}{42mm}\end{minipage}\hfill
\begin{minipage}{0.48\linewidth}\figph{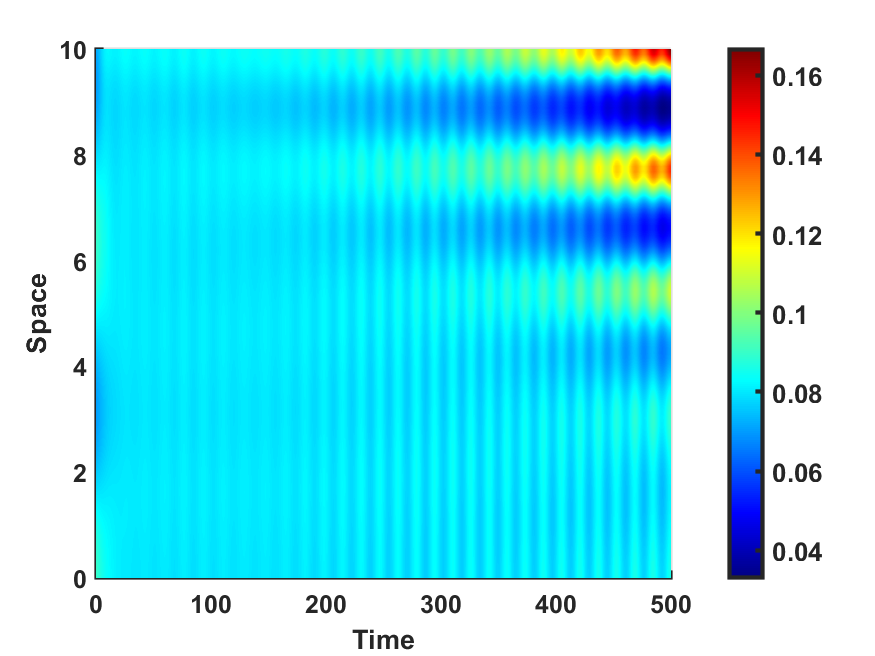}{42mm}\end{minipage}\\[6pt]
\begin{minipage}{0.48\linewidth}\figph{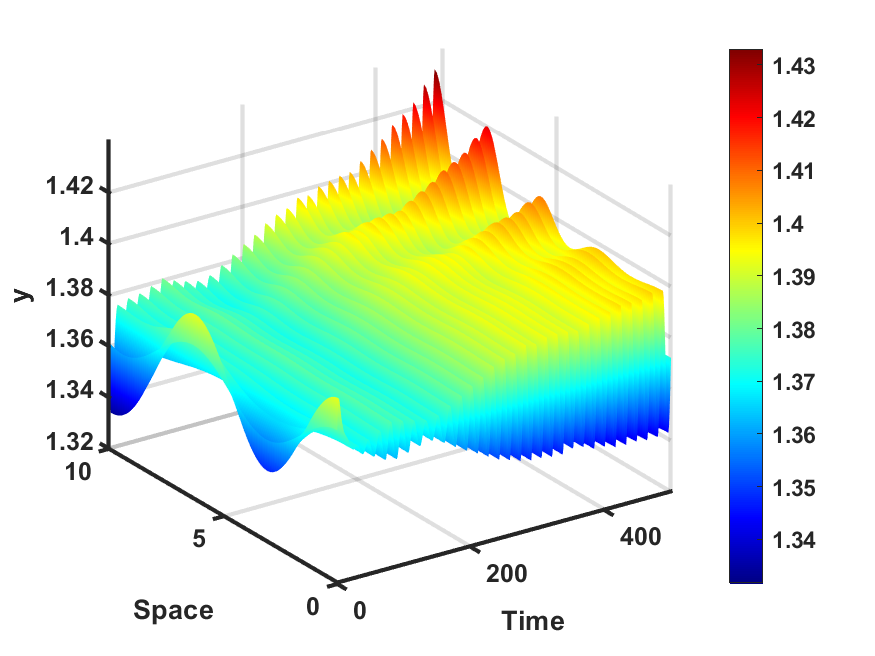}{42mm}\end{minipage}\hfill
\begin{minipage}{0.48\linewidth}\figph{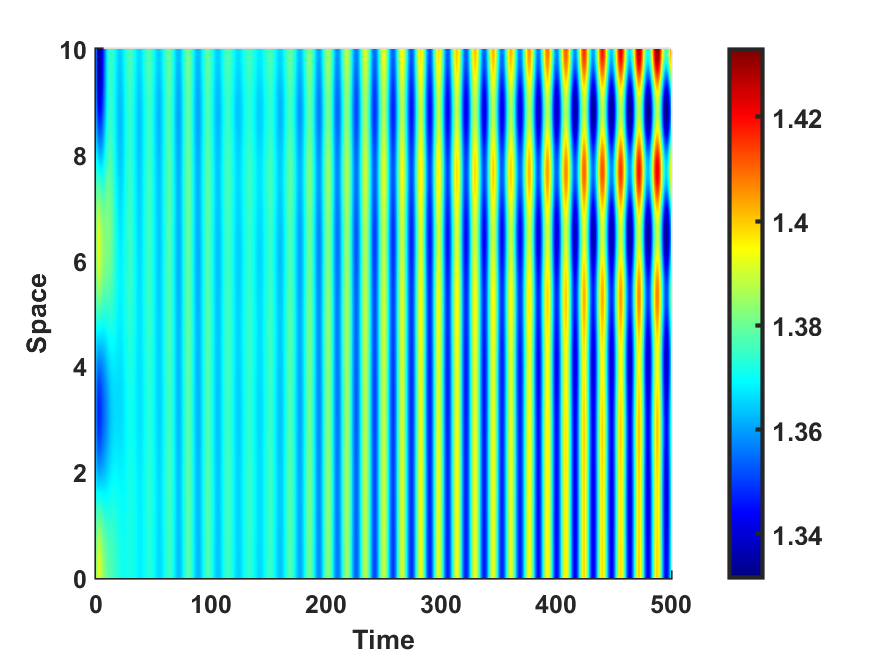}{42mm}\end{minipage}
\caption{\label{fig:9}%
As in Fig.~\ref{fig:7} with initial data
$\bigl(0.0800+0.01\cos u,\,1.3700+0.01\cos u\bigr)$.}
\alttext{Four panels of the same type as the previous two figures, started from a small perturbation of the constant state. The pattern grows out of the small disturbance and then oscillates, showing the state is reached without large initial structure.}
\end{figure}

%% =====================================================================
\section{Discussion}\label{sec:discussion}
%% =====================================================================

The self-limitation term entered the additional-food literature because the model
does not otherwise make sense: without it, solutions run to infinity in finite
time~\cite{parshad2023additional}. What the spatial analysis adds is that the
same term decides which kind of community one actually sees. Three regimes meet
at a single point of the $(d_{1},c)$-plane. When competition is weak, $c<c^{*}$,
the coexistence state is unstable to the uniform mode and the whole habitat rises
and falls together; because the Hopf bifurcation is supercritical
(Sec.~\ref{sec:hopfpdenum}), the cycle grows continuously out of the steady
state, with amplitude scaling as $(c^{*}-c)^{1/2}$, rather than appearing
full-sized. When competition is strong, $c>c^{*}$, and the two species move at
comparable speeds, the community is well mixed and sits still. Between these two,
if the prey disperses much more slowly than the predator, the uniform state gives
way to a finite-wavelength mode and the community freezes into a fixed mosaic of
crowded and empty patches. That pitchfork came out supercritical in every case we
computed, so the mosaic also appears gradually and at small amplitude instead of
by a jump. Close to $(d_{1}^{\mathrm{TH}},c^{*})$ the two mechanisms are equally
strong and the system can do both at once: our integration at $c=c^{*}$,
$d_{1}=0.0030$ settles on a $\cos(5u/2)$ pattern that keeps oscillating in time,
which is a state neither mechanism produces by itself.

The pattern-forming mechanism is the usual activator--inhibitor one, with the
supplement altering both halves of it. The prey is the activator, reinforcing
itself through logistic growth when it is rare; the predator is the inhibitor,
and it operates over a much wider area, since
$d_{1}^{\mathrm{TH}}/d_{2}\approx3\times10^{-3}$ at the threshold. Additional
food pushes in two directions at the same time. Through the term
$\beta\xi/(1+\alpha\xi+x)$ it raises the predator's baseline growth rate and so
strengthens inhibition; by sitting in the denominator of the functional response
it uses up handling time and makes the predator worse at following local
concentrations of prey. The result shows up in where the coexistence state lies:
$\xs=0.064$, under one per cent of the carrying capacity. The predator is being
kept alive almost entirely by the supplement, and the prey survives at a very low
but stubbornly patchy density. This is worth a word of caution for augmentative
biological control. Judged by mean density, a well-mixed model would call the
pest suppressed to a negligible level; the spatial model says the same parameter
values can leave a standing set of refuge patches in which prey density is an
order of magnitude above that mean. Those patches do not wash out with time,
because the instability itself maintains them, not any underlying heterogeneity
of the habitat. Comparable conclusions have been drawn for other supplemented
systems~\cite{ghorai2016pattern}.

The role of the competition coefficient is simple to state and, we think, open to
test. Raising $c$ steadies the population in time: $\mathrm{Re}\,\lambda_{H}$
falls at a rate $b_{2}/2=-0.203$ per unit of $c$, and above $c^{*}$ the uniform
oscillation is gone. It does very little to the spatial side, where
$\partial\lambda_{T}/\partial c=0.0014$ against
$\partial\lambda_{T}/\partial d_{1}=-6.27$, so along the line $c=c^{*}$ the
system stays close to the Turing boundary whatever $c$ does. Releasing natural
enemies that interfere strongly with one another therefore buys stability in time
at almost no cost in space, which is not at all the same as releasing more of
them. The lopsidedness of those two sensitivities also says which variable a
practitioner is really working with: it is the contrast in mobility between pest
and enemy, which is settled by the choice of enemy, and not the amount of food
put out.

Several limits of the analysis should be stated plainly. The habitat is a line
with no flux at the ends, so the available wavenumbers are discrete and $k^{*}$
depends on the length of the domain through $l$. In two dimensions the same
linear calculation picks out a whole circle in Fourier space, and the choice
between stripes, spots and hexagons is then made by quadratic terms that vanish
identically in the one-dimensional Neumann problem; carrying the normal form over
to a planar domain~\cite{jiang2020formulation,song2016turing} is the obvious next
step. The model is deterministic, and it treats the supplement as evenly spread
and inexhaustible. Dropping either assumption is more than a perturbation: a
supplement that can be depleted brings back a third dynamical variable and
reopens the blow-up question that made the self-limitation term necessary in the
first place. Finally, everything here is local. The normal form \eqref{eq:thnf}
governs a neighbourhood of the codimension-two point, and it says nothing about
what happens far from it, such as the mode-$2$-dominated state we found at
$c=1.1$, $d_{1}=0.0087$ with three modes unstable at once. Which mode wins beyond
threshold is a nonlinear question, and often a global one.

Several of the quantities that decide the outcome, including the sign of the
Lyapunov number, the sign of $\mathrm{Re}(b_{021})$ and the position of
$d_{1}^{\mathrm{TH}}$, are differences between terms of similar size, and are
therefore easy to get wrong. Each has been computed at least twice by routes that
share no intermediate steps, and every qualitative prediction has been checked
against direct integration of \eqref{eq:pde}. Those routes agree to between ten
and fourteen significant digits, and that agreement is the basis for the values
reported here.

%% =====================================================================
\section{Conclusion}\label{sec:conclusion}
%% =====================================================================

We have studied a diffusive predator--prey model in which the predator is given
additional food of quantity $\xi$ and quality $1/\alpha$, takes prey through a
Holling type-II response modified by that food, and suffers intraspecific
competition of strength $c\xi$. For the kinetics, the Hopf condition collapses to
a single scalar equation that does not involve $c$, which pins the bifurcation
down exactly, and three separate computations show it to be supercritical. For
the spatial problem we gave necessary and sufficient conditions for
diffusion-driven instability, derived the normal form of the pitchfork that
follows and settled the convention fixing its criticality, and proved that the
spatially uniform Hopf bifurcation is supercritical as well. Taking $(d_{1},c)$
as the two free parameters, the Turing and Hopf curves cross transversally at
$(d_{1},c)=(0.0031426255,\,0.1004263750)$, with critical wavenumber $k^{*}=5$ on
$\Om=(0,2\pi)$; we computed the three-dimensional normal form on the center
manifold there and checked the non-degeneracy conditions for a Hopf--pitchfork
bifurcation. Integrating the partial differential equations directly reproduces
the predicted sequence of states, including the patterned, time-periodic
solutions that live only near the codimension-two point.

The message is that predator self-limitation, put into the additional-food
framework to stop finite-time blow-up, is what separates a community that
oscillates as a whole from one that breaks up in space. The two thresholds
respond very differently to the parameters, sharply to the ratio of diffusivities
and barely at all to the strength of competition. A supplemented system can
therefore be made quiet in time while sitting right at the edge of pattern
formation, so that a pest judged suppressed from mean densities may still be
holding persistent high-density refuges.

The natural extensions are planar domains, where quadratic resonances decide
between stripes and spots; a supplement that is depletable or unevenly
distributed; maturation and gestation delays, which are known to interact
strongly with the Turing--Hopf
structure~\cite{an2019spatiotemporal,jiang2019turing,dai2021turing,li2023bifurcation};
non-local or taxis-driven movement of the
predator~\cite{lv2024turing,luo2021global}; and stochastic forcing, which near a
codimension-two point can throw up large transient patterns well before the
deterministic threshold is reached.

%% =====================================================================
\section*{Declaration of competing interest}
The authors declare that they have no known competing financial interests or personal relationships that could have appeared to influence the work reported in this paper.

\printcredits

\section*{Funding}
This research did not receive any specific grant from funding agencies in the
public, commercial, or not-for-profit sectors.

\section*{Data availability}
The data that support the findings of this study are available within the
article.

%\section*{Acknowledgements}
%The authors appreciate the insightful feedback provided by the Associate Editor
%and the anonymous reviewers.

\bibliographystyle{cas-model2-names}
\bibliography{new_ref}

@book{lotka1925elements,
  author    = {Lotka, Alfred J.},
  title     = {Elements of Physical Biology},
  publisher = {Williams and Wilkins},
  address   = {Baltimore},
  year      = {1925}
}

@article{volterra1926variazioni,
  author  = {Volterra, Vito},
  title   = {Variazioni e fluttuazioni del numero d'individui in specie animali conviventi},
  journal = {Memorie della R. Accademia Nazionale dei Lincei},
  volume  = {2},
  pages   = {31--113},
  year    = {1926}
}

@article{holling1959components,
  author  = {Holling, C. S.},
  title   = {Some characteristics of simple types of predation and parasitism},
  journal = {The Canadian Entomologist},
  volume  = {91},
  number  = {7},
  pages   = {385--398},
  year    = {1959},
  doi     = {10.4039/Ent91385-7}
}

@article{holling1965functional,
  author  = {Holling, C. S.},
  title   = {The functional response of predators to prey density and its role in mimicry and population regulation},
  journal = {Memoirs of the Entomological Society of Canada},
  volume  = {97},
  number  = {S45},
  pages   = {5--60},
  year    = {1965},
  doi     = {10.4039/entm9745fv}
}

@book{freedman1980deterministic,
  author    = {Freedman, H. I.},
  title     = {Deterministic Mathematical Models in Population Ecology},
  publisher = {Marcel Dekker},
  address   = {New York},
  year      = {1980}
}

@book{may2001stability,
  author    = {May, Robert M.},
  title     = {Stability and Complexity in Model Ecosystems},
  publisher = {Princeton University Press},
  address   = {Princeton, NJ},
  year      = {2001}
}

@book{murray2003mathematical,
  author    = {Murray, James D.},
  title     = {Mathematical Biology II: Spatial Models and Biomedical Applications},
  edition   = {3},
  publisher = {Springer},
  address   = {New York},
  year      = {2003}
}

@book{okubo2001diffusion,
  author    = {Okubo, Akira and Levin, Simon A.},
  title     = {Diffusion and Ecological Problems: Modern Perspectives},
  edition   = {2},
  publisher = {Springer},
  address   = {New York},
  year      = {2001}
}

@article{beddington1975mutual,
  author  = {Beddington, J. R.},
  title   = {Mutual interference between parasites or predators and its effect on searching efficiency},
  journal = {Journal of Animal Ecology},
  volume  = {44},
  number  = {1},
  pages   = {331--340},
  year    = {1975},
  doi     = {10.2307/3866}
}

@article{deangelis1975model,
  author  = {DeAngelis, D. L. and Goldstein, R. A. and O'Neill, R. V.},
  title   = {A model for trophic interaction},
  journal = {Ecology},
  volume  = {56},
  number  = {4},
  pages   = {881--892},
  year    = {1975},
  doi     = {10.2307/1936298}
}

@article{srinivasu2007biological,
  author  = {Srinivasu, P. D. N. and Prasad, B. S. R. V. and Venkatesulu, M.},
  title   = {Biological control through provision of additional food to predators: A theoretical study},
  journal = {Theoretical Population Biology},
  volume  = {72},
  number  = {1},
  pages   = {111--120},
  year    = {2007},
  doi     = {10.1016/j.tpb.2007.03.011}
}

@article{srinivasu2010time,
  author  = {Srinivasu, P. D. N. and Prasad, B. S. R. V.},
  title   = {Time optimal control of an additional food provided predator--prey system with applications to pest management and biological conservation},
  journal = {Journal of Mathematical Biology},
  volume  = {60},
  number  = {4},
  pages   = {591--613},
  year    = {2010},
  doi     = {10.1007/s00285-009-0279-2}
}

@article{srinivasu2011role,
  author  = {Srinivasu, P. D. N. and Prasad, B. S. R. V.},
  title   = {Role of quantity of additional food to predators as a control in predator--prey systems with relevance to pest management and biological conservation},
  journal = {Bulletin of Mathematical Biology},
  volume  = {73},
  number  = {10},
  pages   = {2249--2276},
  year    = {2011},
  doi     = {10.1007/s11538-010-9601-9}
}

@article{prasad2013dynamics,
  author  = {Prasad, B. S. R. V. and Banerjee, Malay and Srinivasu, P. D. N.},
  title   = {Dynamics of additional food provided predator--prey system with mutually interfering predators},
  journal = {Mathematical Biosciences},
  volume  = {246},
  number  = {1},
  pages   = {176--190},
  year    = {2013},
  doi     = {10.1016/j.mbs.2013.08.013}
}

@article{sen2015global,
  author  = {Sen, Moitri and Srinivasu, P. D. N. and Banerjee, Malay},
  title   = {Global dynamics of an additional food provided predator--prey system with constant harvest in predators},
  journal = {Applied Mathematics and Computation},
  volume  = {250},
  pages   = {193--211},
  year    = {2015},
  doi     = {10.1016/j.amc.2014.10.085}
}

@article{gurubilli2017global,
  author  = {Gurubilli, Kaushik Kumar and Srinivasu, P. D. N. and Banerjee, Malay},
  title   = {Global dynamics of a prey-predator model with {A}llee effect and additional food for the predators},
  journal = {International Journal of Dynamics and Control},
  volume  = {5},
  number  = {3},
  pages   = {903--916},
  year    = {2017},
  doi     = {10.1007/s40435-016-0234-1}
}

@article{srinivasu2018additional,
  author  = {Srinivasu, P. D. N. and Vamsi, D. K. K. and Ananth, V. S.},
  title   = {Additional food supplements as a tool for biological conservation of predator-prey systems involving type {III} functional response: A qualitative and quantitative investigation},
  journal = {Journal of Theoretical Biology},
  volume  = {455},
  pages   = {303--318},
  year    = {2018},
  doi     = {10.1016/j.jtbi.2018.07.019}
}

@article{parshad2023additional,
  author  = {Parshad, Rana D. and Wickramasooriya, Sureni and Antwi-Fordjour, Kwadwo and Banerjee, Aniket},
  title   = {Additional food causes predators to explode --- unless the predators compete},
  journal = {International Journal of Bifurcation and Chaos},
  volume  = {33},
  number  = {3},
  pages   = {2350034},
  year    = {2023},
  doi     = {10.1142/S0218127423500347}
}

@article{ghorai2016pattern,
  author  = {Ghorai, Santu and Poria, Swarup},
  title   = {Pattern formation and control of spatiotemporal chaos in a reaction diffusion prey--predator system supplying additional food},
  journal = {Chaos, Solitons \& Fractals},
  volume  = {85},
  pages   = {57--67},
  year    = {2016},
  doi     = {10.1016/j.chaos.2016.01.013}
}

@article{turing1952chemical,
  author  = {Turing, Alan M.},
  title   = {The chemical basis of morphogenesis},
  journal = {Philosophical Transactions of the Royal Society of London. Series B, Biological Sciences},
  volume  = {237},
  number  = {641},
  pages   = {37--72},
  year    = {1952},
  doi     = {10.1098/rstb.1952.0012}
}

@article{segel1972dissipative,
  author  = {Segel, Lee A. and Jackson, Julius L.},
  title   = {Dissipative structure: An explanation and an ecological example},
  journal = {Journal of Theoretical Biology},
  volume  = {37},
  number  = {3},
  pages   = {545--559},
  year    = {1972},
  doi     = {10.1016/0022-5193(72)90090-2}
}

@article{klausmeier1999regular,
  author  = {Klausmeier, Christopher A.},
  title   = {Regular and irregular patterns in semiarid vegetation},
  journal = {Science},
  volume  = {284},
  number  = {5421},
  pages   = {1826--1828},
  year    = {1999},
  doi     = {10.1126/science.284.5421.1826}
}

@article{rietkerk2008regular,
  author  = {Rietkerk, Max and van de Koppel, Johan},
  title   = {Regular pattern formation in real ecosystems},
  journal = {Trends in Ecology \& Evolution},
  volume  = {23},
  number  = {3},
  pages   = {169--175},
  year    = {2008},
  doi     = {10.1016/j.tree.2007.10.013}
}

@article{pringle2017spatial,
  author  = {Pringle, Robert M. and Tarnita, Corina E.},
  title   = {Spatial self-organization of ecosystems: Integrating multiple mechanisms of regular-pattern formation},
  journal = {Annual Review of Entomology},
  volume  = {62},
  number  = {1},
  pages   = {359--377},
  year    = {2017},
  doi     = {10.1146/annurev-ento-031616-035413}
}

@article{medvinsky2002spatiotemporal,
  author  = {Medvinsky, Alexander B. and Petrovskii, Sergei V. and Tikhonova, Irene A. and Malchow, Horst and Li, Bai-Lian},
  title   = {Spatiotemporal complexity of plankton and fish dynamics},
  journal = {SIAM Review},
  volume  = {44},
  number  = {3},
  pages   = {311--370},
  year    = {2002},
  doi     = {10.1137/S0036144502404442}
}

@article{petrovskii1999minimal,
  author  = {Petrovskii, Sergei V. and Malchow, Horst},
  title   = {A minimal model of pattern formation in a prey-predator system},
  journal = {Mathematical and Computer Modelling},
  volume  = {29},
  number  = {8},
  pages   = {49--63},
  year    = {1999},
  doi     = {10.1016/S0895-7177(99)00070-9}
}

@article{petrovskii2002allee,
  author  = {Petrovskii, Sergei V. and Morozov, Andrew Y. and Venturino, Ezio},
  title   = {Allee effect makes possible patchy invasion in a predator--prey system},
  journal = {Ecology Letters},
  volume  = {5},
  number  = {3},
  pages   = {345--352},
  year    = {2002},
  doi     = {10.1046/j.1461-0248.2002.00324.x}
}

@article{morozov2006spatiotemporal,
  author  = {Morozov, Andrew and Petrovskii, Sergei and Li, Bai-Lian},
  title   = {Spatiotemporal complexity of patchy invasion in a predator-prey system with the {A}llee effect},
  journal = {Journal of Theoretical Biology},
  volume  = {238},
  number  = {1},
  pages   = {18--35},
  year    = {2006},
  doi     = {10.1016/j.jtbi.2005.05.021}
}

@article{venturino2013spatiotemporal,
  author  = {Venturino, Ezio and Petrovskii, Sergei},
  title   = {Spatiotemporal behavior of a prey--predator system with a group defense for prey},
  journal = {Ecological Complexity},
  volume  = {14},
  pages   = {37--47},
  year    = {2013},
  doi     = {10.1016/j.ecocom.2013.01.004}
}

@article{hu2015pattern,
  author  = {Hu, Guangping and Li, Xiaoling and Wang, Yuepeng},
  title   = {Pattern formation and spatiotemporal chaos in a reaction--diffusion predator--prey system},
  journal = {Nonlinear Dynamics},
  volume  = {81},
  number  = {1-2},
  pages   = {265--275},
  year    = {2015},
  doi     = {10.1007/s11071-015-1988-2}
}

@article{sambath2018spatiotemporal,
  author  = {Sambath, M. and Balachandran, K. and Guin, L. N.},
  title   = {Spatiotemporal patterns in a predator--prey model with cross-diffusion effect},
  journal = {International Journal of Bifurcation and Chaos},
  volume  = {28},
  number  = {2},
  pages   = {1830004},
  year    = {2018},
  doi     = {10.1142/S0218127418300045}
}

@book{hassard1981theory,
  author    = {Hassard, B. D. and Kazarinoff, N. D. and Wan, Y.-H.},
  title     = {Theory and Applications of Hopf Bifurcation},
  series    = {London Mathematical Society Lecture Note Series},
  volume    = {41},
  publisher = {Cambridge University Press},
  address   = {Cambridge},
  year      = {1981}
}

@book{perko2013differential,
  author    = {Perko, Lawrence},
  title     = {Differential Equations and Dynamical Systems},
  edition   = {3},
  series    = {Texts in Applied Mathematics},
  volume    = {7},
  publisher = {Springer},
  address   = {New York},
  year      = {2001}
}

@book{kuznetsov2004elements,
  author    = {Kuznetsov, Yuri A.},
  title     = {Elements of Applied Bifurcation Theory},
  edition   = {3},
  series    = {Applied Mathematical Sciences},
  volume    = {112},
  publisher = {Springer},
  address   = {New York},
  year      = {2004}
}

@book{guckenheimer1983nonlinear,
  author    = {Guckenheimer, John and Holmes, Philip},
  title     = {Nonlinear Oscillations, Dynamical Systems, and Bifurcations of Vector Fields},
  series    = {Applied Mathematical Sciences},
  volume    = {42},
  publisher = {Springer},
  address   = {New York},
  year      = {1983}
}

@book{wu1996theory,
  author    = {Wu, Jianhong},
  title     = {Theory and Applications of Partial Functional Differential Equations},
  series    = {Applied Mathematical Sciences},
  volume    = {119},
  publisher = {Springer},
  address   = {New York},
  year      = {1996}
}

@article{yi2009bifurcation,
  author  = {Yi, Fengqi and Wei, Junjie and Shi, Junping},
  title   = {Bifurcation and spatiotemporal patterns in a homogeneous diffusive predator--prey system},
  journal = {Journal of Differential Equations},
  volume  = {246},
  number  = {5},
  pages   = {1944--1977},
  year    = {2009},
  doi     = {10.1016/j.jde.2008.10.024}
}

@article{jiang2020formulation,
  author  = {Jiang, Weihua and An, Qi and Shi, Junping},
  title   = {Formulation of the normal form of {Turing--Hopf} bifurcation in partial functional differential equations},
  journal = {Journal of Differential Equations},
  volume  = {268},
  number  = {10},
  pages   = {6067--6102},
  year    = {2020},
  doi     = {10.1016/j.jde.2019.11.039}
}

@article{song2016turing,
  author  = {Song, Yongli and Zhang, Tonghua and Peng, Yahong},
  title   = {{Turing--Hopf} bifurcation in the reaction--diffusion equations and its applications},
  journal = {Communications in Nonlinear Science and Numerical Simulation},
  volume  = {33},
  pages   = {229--258},
  year    = {2016},
  doi     = {10.1016/j.cnsns.2015.10.002}
}

@article{an2019spatiotemporal,
  author  = {An, Qi and Jiang, Weihua},
  title   = {Spatiotemporal attractors generated by the {Turing--Hopf} bifurcation in a time-delayed reaction-diffusion system},
  journal = {Discrete and Continuous Dynamical Systems - B},
  volume  = {24},
  number  = {2},
  pages   = {487--510},
  year    = {2019},
  doi     = {10.3934/dcdsb.2018183}
}

@article{an2018turing,
  author  = {An, Qi and Jiang, Weihua},
  title   = {{Turing--Hopf} bifurcation and spatio-temporal patterns of a ratio-dependent {Holling--Tanner} model with diffusion},
  journal = {International Journal of Bifurcation and Chaos},
  volume  = {28},
  number  = {9},
  pages   = {1850108},
  year    = {2018},
  doi     = {10.1142/S0218127418501080}
}

@article{cao2018turing,
  author  = {Cao, Xun and Jiang, Weihua},
  title   = {{Turing--Hopf} bifurcation and spatiotemporal patterns in a diffusive predator--prey system with {Crowley--Martin} functional response},
  journal = {Nonlinear Analysis: Real World Applications},
  volume  = {43},
  pages   = {428--450},
  year    = {2018},
  doi     = {10.1016/j.nonrwa.2018.03.010}
}

@article{chen2019turing,
  author  = {Chen, Xiaoying and Jiang, Weihua},
  title   = {{Turing--Hopf} bifurcation and multi-stable spatio-temporal patterns in the {Lengyel--Epstein} system},
  journal = {Nonlinear Analysis: Real World Applications},
  volume  = {49},
  pages   = {386--404},
  year    = {2019},
  doi     = {10.1016/j.nonrwa.2019.03.013}
}

@article{jiang2019turing,
  author  = {Jiang, Weihua and Wang, Hongbin and Cao, Xun},
  title   = {{Turing} instability and {Turing--Hopf} bifurcation in diffusive {Schnakenberg} systems with gene expression time delay},
  journal = {Journal of Dynamics and Differential Equations},
  volume  = {31},
  number  = {4},
  pages   = {2223--2247},
  year    = {2019},
  doi     = {10.1007/s10884-018-9702-y}
}

@article{dai2021turing,
  author  = {Dai, Binxiang and Sun, Guangxun},
  title   = {{Turing--Hopf} bifurcation of a delayed diffusive predator--prey system with chemotaxis and fear effect},
  journal = {Applied Mathematics Letters},
  volume  = {111},
  pages   = {106644},
  year    = {2021},
  doi     = {10.1016/j.aml.2020.106644}
}

@article{jiang2019schooling,
  author  = {Jiang, Heping},
  title   = {{Turing} bifurcation in a diffusive predator-prey model with schooling behavior},
  journal = {Applied Mathematics Letters},
  volume  = {96},
  pages   = {230--235},
  year    = {2019},
  doi     = {10.1016/j.aml.2019.05.010}
}

@article{li2013hopf,
  author  = {Li, Xin and Jiang, Weihua and Shi, Junping},
  title   = {{Hopf} bifurcation and {Turing} instability in the reaction--diffusion {Holling--Tanner} predator--prey model},
  journal = {IMA Journal of Applied Mathematics},
  volume  = {78},
  number  = {2},
  pages   = {287--306},
  year    = {2013},
  doi     = {10.1093/imamat/hxr050}
}

@article{li2023bifurcation,
  author  = {Li, Peiluan and Gao, Rong and Xu, Changjin and Ahmad, Shabir and Li, Ying and Akg\"{u}l, Ali},
  title   = {Bifurcation behavior and $PD^{\gamma}$ control mechanism of a fractional delayed genetic regulatory model},
  journal = {Chaos, Solitons \& Fractals},
  volume  = {168},
  pages   = {113219},
  year    = {2023},
  doi     = {10.1016/j.chaos.2023.113219}
}

@article{fu2025turing,
  author  = {Fu, Xinyu and Jiang, Heping},
  title   = {{Turing--Hopf} bifurcation in a diffusive predator--prey model with schooling behavior and {Smith} growth},
  journal = {Applied Mathematics Letters},
  volume  = {159},
  pages   = {109257},
  year    = {2025},
  doi     = {10.1016/j.aml.2024.109257}
}

@article{li2025turing,
  author  = {Li, Lu and Yan, Xiang-Ping and Zhang, Cun-Hua},
  title   = {{Turing}, {Hopf} and {Turing--Hopf} bifurcations in a modified {Leslie--Gower} predator--prey diffusive system with {Smith} prey growth and nonmonotonic functional response},
  journal = {Chaos, Solitons \& Fractals},
  volume  = {201},
  pages   = {117226},
  year    = {2025},
  doi     = {10.1016/j.chaos.2025.117226}
}

@article{li2023turing,
  author  = {Li, Yanqiu and Zhou, Yibo},
  title   = {{Turing--Hopf} bifurcation in a general {Selkov--Schnakenberg} reaction--diffusion system},
  journal = {Chaos, Solitons \& Fractals},
  volume  = {171},
  pages   = {113473},
  year    = {2023},
  doi     = {10.1016/j.chaos.2023.113473}
}

@article{lv2024turing,
  author  = {Lv, Yehu},
  title   = {{Turing--Hopf} bifurcation analysis and normal form in delayed diffusive predator--prey system with taxis and fear effect},
  journal = {Journal of Applied Mathematics and Computing},
  volume  = {70},
  number  = {6},
  pages   = {5721--5761},
  year    = {2024},
  doi     = {10.1007/s12190-024-02183-4}
}

@article{luo2021global,
  author  = {Luo, Dongpo},
  title   = {Global bifurcation for a reaction--diffusion predator--prey model with {Holling-II} functional response and prey--taxis},
  journal = {Chaos, Solitons \& Fractals},
  volume  = {147},
  pages   = {110975},
  year    = {2021},
  doi     = {10.1016/j.chaos.2021.110975}
}

@article{tao2021study,
  author  = {Tao, Xiangyu and Zhu, Linhe},
  title   = {Study of periodic diffusion and time delay induced spatiotemporal patterns in a predator-prey system},
  journal = {Chaos, Solitons \& Fractals},
  volume  = {150},
  pages   = {111101},
  year    = {2021},
  doi     = {10.1016/j.chaos.2021.111101}
}

@article{bhunia2023study,
  author  = {Bhunia, Bidhan and Ghorai, Santu and Kar, Tapan Kumar and Biswas, Samir and Bhutia, Lakpa Thendup and Debnath, Papiya},
  title   = {A study of a spatiotemporal delayed predator--prey model with prey harvesting: Constant and periodic diffusion},
  journal = {Chaos, Solitons \& Fractals},
  volume  = {175},
  pages   = {113967},
  year    = {2023},
  doi     = {10.1016/j.chaos.2023.113967}
}

@article{li2025spatiotemporal,
  author  = {Li, Yanqiu and Sun, Jia},
  title   = {Spatiotemporal dynamics of a delayed diffusive predator--prey model with hunting cooperation in predator and anti-predator behaviors in prey},
  journal = {Chaos, Solitons \& Fractals},
  volume  = {198},
  pages   = {116561},
  year    = {2025},
  doi     = {10.1016/j.chaos.2025.116561}
}

@article{chen2023turing,
  author  = {Chen, Mengxin and Xu, Yuanyuan and Zhao, Jinbiao and Wei, Xin},
  title   = {{Turing--Hopf} bifurcation analysis in a diffusive ratio-dependent predator--prey model with {Allee} effect and predator harvesting},
  journal = {Entropy},
  volume  = {26},
  number  = {1},
  pages   = {18},
  year    = {2023},
  doi     = {10.3390/e26010018}
}

@article{miao2022hopf,
  author  = {Miao, Liangying and He, Zhiqian},
  title   = {{Hopf} bifurcation and {Turing} instability in a diffusive predator-prey model with hunting cooperation},
  journal = {Open Mathematics},
  volume  = {20},
  number  = {1},
  pages   = {986--997},
  year    = {2022},
  doi     = {10.1515/math-2022-0474}
}

@article{toab218,
    author = {Banerjee, Aniket and Valmorbida, Ivair and O’Neal, Matthew E and Parshad, Rana},
    title = {Exploring the Dynamics of Virulent and Avirulent Aphids: A Case for a ‘Within Plant’ Refuge},
    journal = {Journal of Economic Entomology},
    volume = {115},
    number = {1},
    pages = {279-288},
    year = {2022},
    month = {02},
    issn = {0022-0493},
    doi = {10.1093/jee/toab218}
}

@Article{two_aniket,
title = {Two species competition with a "non-smooth" Allee mechanism: applications to soybean aphid population dynamics under climate change},
journal = {Mathematical Biosciences and Engineering},
volume = {22},
number = {3},
pages = {604-651},
year = {2025},
issn = {1551-0018},
doi = {10.3934/mbe.2025023},
url = {https://www.aimspress.com/article/doi/10.3934/mbe.2025023},
author = {Aniket Banerjee and Urvashi Verma and Margaret T. Lewis and Rana D. Parshad}}

@article{verma2026t,
  title={T (w) o Patch or Not T (w) o Patch: A Novel Biocontrol Model},
  author={Verma, Urvashi and Banerjee, Aniket and Parshad, Rana D},
  journal={Mathematical Methods in the Applied Sciences},
  year={2026},
  publisher={Wiley Online Library}
}

\end{document}